\documentclass[12pt,draftcls,onecolumn]{IEEEtran}
\usepackage{amsmath}
\usepackage{amsthm}
\usepackage{amsfonts}
\usepackage{amssymb}
\usepackage[utf8]{inputenc}
\DeclareUnicodeCharacter{00A0}{ }
\usepackage{graphicx}
\usepackage{graphics}
\usepackage{float}
\usepackage{tikz}
\usetikzlibrary{shapes,snakes,patterns}
\usepackage{epstopdf}

\usepackage[justification=centering]{caption}
\usepackage{enumerate} 
\usepackage{cite}
\usepackage{placeins}
\newcommand{\vast}{\bBigg@{4}}
\newcommand{\Vast}{\bBigg@{5}}
\newtheorem{theorem}{Theorem}[]
\newtheorem{lemma}[]{Lemma}
\usepackage{suffix}
\usepackage{mathtools}
\usepackage{amsthm}
\usepackage{enumitem}
\usepackage{xfrac}
\usepackage{relsize}
\usepackage{hyperref}
\theoremstyle{definition}
\newtheorem{defn}{Definition}[]

\usepackage[list=true]{subcaption}
	\title{Asymptotic maximum order statistic for SIR in $\kappa-\mu$ shadowed fading}
\author{Athira Subhash, Muralikrishnan Srinivasan, Sheetal Kalyani \\
	\hspace{-0.5 cm}Department of Electrical Engineering,\\
	\hspace{-0.8 cm} Indian Institute of Technology, Madras, \\
	\hspace{-1cm} Chennai, India 600036.\\
	\hspace{-1cm} \{ee16d027@smail,ee14d206@ee,skalyani@ee\}.iitm.ac.in\\
\thanks{The authors are with the Department of Electrical Engineering, Indian Institute of Technology Madras,  Chennai, India 600036 (email:\{ee14d206@ee,ee16d027@smail,skalyani@ee\}.iitm.ac.in). }
\thanks{Copyright(c) 2019 IEEE. Personal use is permitted. For any other purposes, permission must be obtained from the IEEE by emailing pubs-permissions@ieee.org. This is the author’s version of an article that has been published in this journal. Changes were made to this version by the publisher prior to publication. The final version of record is  \href{https://ieeexplore.ieee.org/abstract/document/8745482}{available here}}
}
\begin{document}
\bstctlcite{IEEEexample:BSTcontrol}

\maketitle
\vspace{-2mm}
\begin{abstract}
	Using tools from extreme value theory (EVT), it is proved that, when the user signal and the interferer signals undergo independent and non-identically distributed (i.n.i.d.) $\kappa-\mu$ shadowed fading, the limiting distribution of the maximum of $L$ independent and identically distributed (i.i.d.) signal-to-interference ratio (SIR) random variables (RVs) is a Frechet distribution. It is observed that this limiting distribution is close to the true distribution of maximum, for maximum SIR evaluated over moderate $L$. Further, moments of the maximum RV is shown to converge to the moments of the Frechet RV. Also, the rate of convergence of the actual distribution of the maximum to the Frechet distribution is derived and is analyzed for different $\kappa$ and $\mu$ parameters. Finally,  results from stochastic ordering are used to analyze the variation in the limiting distribution with respect to the variation in source fading parameters. These results are then used to derive upper bound for the rate in Full Array Selection (FAS) schemes for antenna selection and the asymptotic outage probability and the ergodic rate in maximum-sum-capacity (MSC) scheduling systems.
\end{abstract}
\vspace{2cm}

	\section{Introduction}
	\par Massive multiple input multiple output (MIMO) system has been widely accepted as a key to meet the increasing demand for wireless throughput in 5G systems \cite{marzetta2015massive}. With the deployment of massive MIMO, one can expect transmitters/receivers with hundreds of antennas available for communication simultaneously. Larsson et al. in \cite{larsson2014massive} show that the uplink spectral efficiency and the radiated power efficiency shall increase by 100 times with massive MIMO technology (with 100 antennas) and appropriate signal processing techniques. There are several examples in literature in which massive MIMO scenarios have more than 100 antennas. Works such as \cite{vieira2014flexible,guan2017flexible} present various simulation results for realizing massive MIMO in a practical test-bed. Similarly, \cite{gao2013antenna, gao2015massive,bjornson2016massive,hanif2018antenna, gao2016antenna}  present analysis of different massive MIMO systems with over 100 antennas.
	\par One of the impediments to a dense deployment of cellular networks, especially in MIMO systems is the co-channel interference (CCI), which is caused by sharing of common system resources by multiple users and by frequency reuse among adjacent cells. Therefore, the effect of CCI on the quality of the wireless link has to be studied extensively before cell-planning and employing interference mitigation techniques. A vast amount of attention and research in literature is devoted to the study of signal to interference ratio (SIR). Given the fact that massive MIMO is a promising technology for future cellular scenarios, the evaluation of the maximum SIR statistics over all available antennas will be a useful metric for various performance analysis and other quality of service (QoS) provisioning applications. Very recently, \cite{gao2018massive} discussed bounds on the rate of full antenna selection (FAS) architecture in a massive MIMO system, using statistics of the maximum SIR in a Rayleigh fading scenario. Similarly, the maximum SIR is an important performance metric in multi-user shared networks, where user scheduling is based on the channel conditions of the users. For example, the authors of \cite{sibomana2016ergodic} derive analytical expressions of the ergodic capacity for max-signal-to-interference-plus-noise-ratio (Max-SINR) scheduling system in a cognitive radio network. Also, using the $k$th order statistics of the user SIR, \cite{al2018asymptotic} analyzes the asymptotic performance of a generalized multi-user diversity scheme of an underlay cognitive radio system in a Nakagami fading channel.
	\par The cumulative distribution function (CDF) of the maximum of independent random variables (RVs) is given by the product of CDF of each of the variables.  Hence, in the case of $L$ independent and identically distributed (i.i.d.) RV's, the CDF of the maximum is given by the $L$th power of the common CDF. In several cases, given that the CDF of a single RV can itself involve complicated functions, the CDF of order statistics like maximum and minimum, even over i.i.d. RVs will be more complicated. Also, providing a meaningful analysis for performance metrics like outage probability, the ergodic rate becomes intractable, if not impossible. In such cases, we can use extreme value theory (EVT) and propose a systematic approach to characterize the asymptotic maximum or minimum SIR in terms of simple probability distribution functions (PDF) or CDFs that are amenable to analysis. For example, works like \cite{jindal2006multicast, park2008multicast, park2009multicast} study the capacity limits of Rayleigh faded multi-cast channels using EVT, without which the capacity limits would have been intractable to analyze. EVT has also been used effectively for studying the asymptotic behaviour of performance metrics in opportunistic scheduling. For example, the limiting distribution of spectral efficiency for multi-hop relaying techniques employing opportunistic scheduling is analyzed in works such as \cite{oyman2007scheduling,oyman2008scheduling, oyman2010scheduling} using EVT. The ergodic capacity of opportunistic scheduling for a gamma-gamma composite fading channel is investigated in the work \cite{ahmadi2012scheduling}. The asymptotic distributions of metrics such as ergodic capacity, mutual information, end-to-end signal to noise ratio (SNR), ergodic secrecy rate (ESR) in a multi-relay setup are discussed in works such as \cite{kountouris2009scheduling, xue2010mi, xia2014scheduling, biswas2016relay, xu2016relay,sibomana2016ergodic,al2018asymptotic}. The asymptotic PDF of the maximum of i.i.d. sums of i.n.i.d. gamma RVs is shown to be a Gumbel PDF  in \cite{kalyani2012gamma}. SIR-based asymptotic throughput analysis for opportunistic scheduling of MIMO downlink systems for Rayleigh fading channels is performed in \cite{pun2011mimo}. Here, using EVT, the limiting distribution is found to be Frechet distribution. To the best of our knowledge, there is currently no work that gives results similar to \cite{pun2011mimo}, even for Rician or Nakagami fading channels. Also, the recent work \cite{gao2018massive}, derives the statistical upper channel capacity bounds for FAS systems using EVT in the large-scale limit only for Rayleigh fading channels. 
	\par In recent times, there has been a significant focus on generalized multipath fading models, first discussed in \cite{yacoub_k_mu}. These fading models called $\kappa-\mu$ and $\eta-\mu$ fading, model small-scale variations of the channel in the line of sight (LOS) and non-line of sight (NLOS) conditions, respectively. Further, these generalized fading distributions include Rayleigh, Rician, Nakagami-m, Nakagami-q and one-sided Gaussian distributions as special cases. To investigate shadowing of the dominant component, a shadowed Rice model with random LOS component is introduced in\cite{abdi_lms}. A further generalization of the shadowed Rician fading is the $\kappa-\mu$ shadowed fading, which has been studied in both \cite{paris2014statistical} and \cite{cotton_d2d}. Also, $\kappa-\mu$ shadowed fading has been shown to unify the $\kappa-\mu$ and $\eta-\mu$ fading models \cite{pozas_shadowed} and to have a wide variety of applications ranging from land-mobile satellite systems to device to device communication \cite{cotton_d2d}. 
	\par Performance metrics for $\kappa-\mu$ shadowed fading have been studied extensively in works like \cite{celia2014capacity, zhang2015effective, chen2016outage, li2017rate, zhang2017hos, chandrasekaran2015performance,thomas2016error}. For example, the exact capacity and effective capacity expressions for $\kappa-\mu$ shadowed fading channel have been derived in \cite{celia2014capacity} and \cite{zhang2015effective} respectively. Expressions for the effective rate of MISO systems over $\kappa-\mu$ shadowed fading models have been derived in \cite{li2017rate}. However, all the above works either do not consider the impact of CCI or consider only Rayleigh faded interferers. There are works like \cite{morales2012outage, paris2013outage, ermolova2014outage, kumar2015coverage, kumar2015outage,zhang2017performance}, which consider CCI in a generalized fading setting and characterize the SIR. For example, outage probability expression for $\eta-\mu$ signal of interest (SOI) and Rayleigh faded interferers is derived in terms of confluent Lauricella function in \cite{morales2012outage}. Outage probability expressions, when SOI experiences $\eta-\mu$ or $\kappa-\mu$ fading and the interfering signals are subject to $\eta-\mu$ fading, have been derived in \cite{paris2013outage}. This was further extended to cases where CCI can be either $\eta-\mu$ or $\kappa-\mu$ fading in the presence of white Gaussian noise in \cite{ermolova2014outage}. Expressions for coverage probability and rate are derived in terms of Lauricella's function of the fourth kind in \cite{kumar2015coverage} when SOI and interferers experience $\kappa-\mu$ and $\eta-\mu$ fading respectively. Approximate outage probability and rate expressions are derived in terms of the Appell function in \cite{kumar2015outage}, when the user channel and the interferers experience $\kappa-\mu$ and $\eta-\mu$ fading respectively. 
	\par Though new, $\kappa-\mu$ shadowed fading has its fair share in the literature that characterize SIR. In \cite{parthasarathy2017coverage} coverage probability expressions are derived when the base stations are modeled as Poisson point process (PPP) and the channels experience $\kappa-\mu$ shadowed fading. Expressions for error vector magnitude (EVM) are derived in \cite{parthasarathy2018evm} for an interference-limited system when both the desired channel and interferers experience i.n.i.d. $\kappa-\mu$ shadowed fading. Approximate outage probability and capacity expressions are derived for $\kappa-\mu$ shadowed fading channels in \cite{kumar2015capacity}. Exact outage and rate expressions in the presence of CCI has been studied in \cite{kumar2017outage} only recently. 
	\par One thing that is common among \cite{chen2016outage}, \cite{morales2012outage, paris2013outage, ermolova2014outage, kumar2015coverage, kumar2015outage, parthasarathy2017coverage, parthasarathy2018evm, kumar2015capacity, kumar2017outage} is the complicated nature of the PDF and the CDF of SIR. For example, the recent work \cite{kumar2017outage}, which generalizes all existing results and considers the SOI and CCI to be i.n.i.d. $\kappa-\mu$ shadowed fading derives the CDF of SIR in terms of an infinite summation of the Lauricella function of the fourth kind. This Lauricella function itself involves N-fold infinite summation (Here, $N$ denotes the number of interferers). Now, determining the CDF of maximum over $L$ such i.i.d. SIR realizations involve raising the CDF to power $L$, making further mathematical analysis like computing rate very difficult.  Even the evaluation of the exact CDF of the maximum of two SIR RVs having two i.n.i.d interferers takes more than an hour to compute in Mathematica with the series expansion given by \cite[Eqn. (8)]{kumar2017outage}. Further, the evaluation of the exact CDF of the maximum of four SIR RVs with each SIR RV having four i.n.i.d interferers in a $\kappa-\mu$ shadowed fading environment times out in Mathematica. Therefore, a limiting distribution for the maximum of SIR RVs, which is not only easy to compute but is also amenable to mathematical analysis, will have significant utility. Also, such a distribution will easily extend the recent FAS results of \cite{gao2018massive} to a generalized fading scenario. Similarly, the authors of \cite{song2006asymptotic} and \cite{sibomana2016ergodic} discuss the performance analysis of a maximum-sum-capacity (MSC) scheduling system and a max-SINR scheduling system respectively in Rayleigh fading channels. A simple expression for the distribution of maximum SIR can generalize these results as well. Our major contributions in this paper are as follows:
	\begin{itemize}
	    \item Assuming that the user signal and the interferer signals undergo i.n.i.d. $\kappa-\mu$ shadowed fading, we prove, using tools from EVT, that the limiting distribution of the maximum of $L$ such i.i.d. SIR RVs is a Frechet distribution. We then prove the convergence of moments of the maximum RV to those of the limiting distribution.
	    \item We also derive the rate of convergence of the actual maximum distribution to the asymptotic distribution. This sheds light on how well the limiting distribution approximates the actual distribution for finite values of $L$ and $N$. In order to further demonstrate the practical validity of the work, we also study the empirical Kullback-Leibler (KL) divergence between the empirical maxima distribution and the derived asymptotic distribution. The KL divergence results indicate the quantitative closeness between the asymptotic results and the exact results, for finite $L$, whereas the rate of convergence results discusses the order of convergence. 
	    \item Further, we use results from stochastic ordering to analyze the variations in the asymptotic distribution of the maximum. This analysis will not be possible with the exact but complicated distribution of the maximum RV.
	    \item Finally, we analyze the utility of the derived asymptotic results in the following applications: 
	    \begin{enumerate}[label=\roman*]
	        \item Analysis of asymptotic outage probability and asymptotic ergodic rate of the user in each time slot of a MSC system.
	        \item Derivation of upper bound on the rate in FAS architectures for antenna selection in massive MIMO scenario.
	    \end{enumerate}
	\end{itemize} 
	Also, the above-mentioned results hold for Rayleigh, Rician, Nakagami-m, $\kappa-\mu$ and $\eta-\mu$ faded user and interferer scenarios since all of these are special cases of the $\kappa-\mu$ shadowed fading model. Since we assume i.n.i.d. interferers, we also account for interferers having different path-loss or having unequal powers. 
	
    \par The rest of the paper is organized in the following fashion. In Section \ref{evt}, we find the asymptotic distribution of the maximum SIR using tools from EVT. We also give brief notes on the convergence of moments and the rate of convergence. Further, in Section \ref{analysis}, we give an analysis of the asymptotic distribution and analyze the convergence of the true maxima distribution to the asymptotic results derived in terms of the empirical KL divergence. Then, in Section \ref{results}, we present three applications of the derived results and their corresponding simulations. Finally, we conclude the work in Section \ref{conclusion}.

\section{EVT based maxima of $L$ i.i.d. SIR RVs}\label{evt}
 Let $\mathlarger{\gamma^{L}}_{max}$ denote the maximum of $L$ i.i.d. SIR RVs, where the source and the interferers are assumed to experience i.n.i.d. $\kappa-\mu$ shadowed fading, i.e., $\gamma^L_{max} = max\{\gamma_1,\cdots,\gamma_L \}$ and $\gamma_j \sim F_\gamma(z), \ \forall \ j \in \{1,\cdots,L\}$. In this section, (a) the asymptotic distribution of $\gamma^{L}_{max}$ represented as ${F}_{{\gamma}^{}_{max}}(z)$ is derived, (b) convergence of the moments of ${F}_{{\gamma}^{}_{max}}(z)$ to the moments of the true maxima distribution $F_{\gamma^{L}_{max}}(z)$ is analyzed and (c) the rate of convergence of ${F}_{\gamma^{L}_{max}}(z)$ to ${F}_{{\gamma}^{}_{max}}(z)$ is derived.  

\subsection{Maximum of SIR RVs in $\kappa-\mu$ shadowed fading environment.} \label{sectionA}
We first prove that the CDF of maximum of $L$ i.i.d. SIR RVs converges to the CDF of a Frechet RV. For this, we make use of Fisher-Tippet theorem, which forms the corner-stone of EVT. The seminal theorem is as follows \cite{de2007extreme}:
\begin{theorem}  \label{fischer_tippet}
	Fisher-Tippet Theorem, Limit Laws for Maxima: \\
	Let $\mathlarger{z_1,z_2,\cdots,z_L}$ be a sequence of $L$ i.i.d. RVs and $\mathlarger{M_L=}$ max $\mathlarger{\{z_1,z_2,\cdots,z_L\} }$; if $\mathlarger{\exists}$ constants $\mathlarger{a_L>0}$ and $\mathlarger{b_L\in \mathbb{R}}$ and some non-degenerate CDF $\mathlarger{G_{\beta}}$ such that, as $\mathlarger{L \to \infty}$,
	\begin{equation}
		a_L^{-1}\left(M_L-b_L\right) \xrightarrow[]{D} G_{\beta},
		\label{condn_exist}
	\end{equation}
	where $\mathlarger{\xrightarrow[]{D}}$ denotes convergence in distribution,  then the CDF $\mathlarger{G_{\beta}}$ is one of the three CDFs:
	\begin{itemize}
		\item [] $Frechet \ : \ \Lambda_1 (z) :=  \begin{cases}
		0, & z \leq 0,  \\
		exp(-z^{-\beta}), & z>0,
		\end{cases} $ 
		\item []  $Reversed \ Weibull \ : \ \Lambda_2 (z) :=  \begin{cases}
		exp(-(-z)^{\beta}), & z \leq 0, \\
		1, & z>0,
		\end{cases} $
		\item [] $Gumbel \ : \ \Lambda_3(z) :=exp(-exp{(-z)}), \ \ \ z \in \mathbb{R}.$
	\end{itemize}
\end{theorem}
\begin{proof}
	Please refer to page 6 in \cite{de2007extreme} for the proof.
\end{proof}
To determine the limiting distribution from the above three, we have to first define the maximum domain of attraction (\textit{MDA}). 
\theoremstyle{definition}
\begin{defn}{Maximum Domain of Attraction \cite{de2007extreme}:} \label{def_mda}
The CDF $F$ of i.i.d. RVs $z_1,\cdots, z_L$ belongs to the $MDA$ of the extreme value distribution (EVD) $G_\beta$, if and only if $\exists$ constants $a_L>0$ and $b_L\in \mathbb{R}$, such that (\ref{condn_exist}) holds.
\end{defn}
\begin{theorem} \label{mda_frechet}
A CDF $F$ belongs to the $MDA$ of the Frechet distribution, if it satisfies the following relation from \cite{de2007extreme}:
	 \begin{equation}
	 	\lim_{t\to\infty}\frac{1-F(tz)}{1-F(t)} \ = \ z^{-\beta}.
	 	\label{codn_frec1}
	 \end{equation}
	 \end{theorem}
	 \begin{proof}
	Please refer to page 19 in \cite{de2007extreme} for the proof.
\end{proof}
Now, if we show that the CDF $\mathlarger{F_{\gamma}(z)}$ satisfies the relation in (\ref{codn_frec1}), then from the definition of the $MDA$ of an EVD, we can conclude that there exists $\mathlarger{a_L}$ and $\mathlarger{b_L}$ satisfying (\ref{condn_exist}). A choice for the corresponding constants for the Frechet distribution is given in \cite{de2007extreme} as $\mathlarger{b_L = 0}$ and $\mathlarger{a_L = F^{-1}(1-L^{-1})}$.
\begin{theorem}
	The CDF $\mathlarger{F_{\gamma}(z)}$ is in the $MDA$ of the Frechet distribution.
\end{theorem}
\begin{proof}
    Please refer to Appendix \ref{proof_mda} for the detailed proof.
\end{proof}
Thus, we conclude that the CDF of $\mathlarger{\gamma^{L}_{max}}$ converges to the CDF of a Frechet RV ${\gamma}_{max}$ with shape parameter 
\begin{equation}
    \beta = \sum\limits_{i=1}^N\mu_i
    \label{beta}
\end{equation}
and scale parameter 
\begin{equation}
    a_L = F_{\gamma}^{-1}(1-L^{-1}). 
    \label{aL}
\end{equation}
The asymptotic distribution of the CDF of $\gamma^{L}_{max}$ is hence given by 
\begin{equation}
   {F}_{ {\gamma}_{max}}(z) =  \begin{cases}
		0, & z \leq 0,  \\
		\exp\left(-\left( {z}/{a_L}\right)^{-\beta}\right), & z>0.
		\end{cases} 
		\label{asymp_cdf}
\end{equation}
The above expression is far easier to evaluate than the $L^{th}$ power of (\ref{cdf1}) for large values of $L$. 

\subsection{Moment Convergence} \label{sectionB}
We will examine the convergence of moments of $\gamma^L_{max}$ to those of ${\gamma}_{max}$. This is useful in evaluating various average ergodic performance metrics with respect to the maximum SIR RV. We make use of the following results from \cite{de2007extreme} to prove the convergence of moments. 
\begin{lemma} \label{moment_lemma}
	If F, the CDF of a RV $Z$, belongs to the domain of attraction of $\mathlarger{G_{\beta}}$, then $\mathlarger{\forall}$ $\mathlarger{-\infty<z<\omega(F)}$,
	\begin{equation}
		\mathbb{E}[|Z|^{\nu} \ \textbf{1}_{Z>z}] :   \begin{cases}
		< \ \infty, & if \ 0<\nu<\beta^+,  \\
		= \ \infty, & if \ \nu > \beta^+,
		\end{cases}
	\end{equation}
	
	where $\mathlarger{\beta^+ := max\{0,\beta\}}$,  $\mathlarger{\omega(F):= sup \{z \in \mathbb{R} : F(z)<1 \}}$ and $\mathlarger{\textbf{1}_{Z>z}}$ is the indicator function for the event given by $\mathlarger{Z>z}$. 
\end{lemma}
	 \begin{proof}
	Please refer to \cite{de2007extreme} for the proof.
\end{proof}
\begin{theorem} \label{moment_thm}
	Let $Z$ be an $F$ distributed RV and $F$ belongs to the domain of attraction of $G_{\beta}$, if $\mathbb{E}[Z^{\nu}]$ is finite for some $\nu<\beta^+$ then,
	\begin{equation}
		\lim_{L\to\infty}\mathbb{E}\left[\left(\frac{M_L-b_L}{a_L}\right)^\nu\right] = \int\limits_{-\infty}^{\infty} z^{\nu} dG_{\beta}(z),
		\label{moment_condtn}
	\end{equation}
where $\beta^+ := max\{0,\beta\}$.
\end{theorem}
	 \begin{proof}
	Please refer to page 176 in \cite{de2007extreme} for the proof.
\end{proof}

From Lemma \ref{moment_lemma}, we observe that $\mathbb{E}[Z^{\nu}]$ is finite for all values of $0 < \nu<\beta^+$. So, according to Theorem \ref{moment_thm}, (\ref{moment_condtn}) holds for all $\nu$ in this range. Hence, we conclude that the $\nu^{th}$ moment of the RV $\mathlarger{\gamma^{L}_{max}}$ converges to the $\mathlarger{\nu^{th}}$ moment of ${\gamma}_{max}$, for all $\mathlarger{\nu \ < \ \sum\limits_{i=1}^{N}\mu_i}$. 
\vspace{2mm}

\subsection{Rate of convergence} \label{sectionC}
Note that, (\ref{condn_exist}) guarantees the convergence of the distribution of $\gamma^{L}_{max}$ to a Frechet distribution, but does not discuss the rate of convergence. In other words, it does not discuss how fast $\gamma^{L}_{max} \xrightarrow[]{D}{\gamma}_{max}$. The rate of convergence is not the same for all distributions in any domain of attraction. In fact, it is a function of the initial distribution parameters and depends on the equivalence of the tail of the initial distribution function to the tail of a generalized Pareto distribution (GPD) \cite{de2007extreme}. The closer the tail-behaviour of the initial distribution to the tail-behaviour of a GPD, faster is its rate of convergence. We now give the rate of convergence for our case through the following theorem.
\begin{theorem} \label{rate_cnvg}
	The rate of convergence of $F_{\gamma^{L}_{max}}(z)$ to the Frechet distribution is \\ $O\left(L^{-\left(\sum\limits_{i=1}^{N}\mu_{i} \right)^{-1}} +  L^{-1} \right)$.
\end{theorem}

\begin{proof}
Please see Appendix. \ref{rateofconv} for the detailed proof. 
\end{proof}
This result is equivalent to stating that 	\begin{equation}
		\sup_{B\in \mathbb{B}} \left| \mathbb{P}\left(\left(\left( \frac{\gamma^{L}_{max}}{a}  \right) / L^{\beta}\right) \in B\right) - \Lambda_1(B)\right| = O\left(\left(\frac{1}{L}\right)^{\delta} + \frac{1}{L}\right),
		\end{equation}
	where $\delta=\left(\sum\limits_{i=1}^{N}\mu_{i} \right)^{-1}$, $\mathbb{B}$ denotes the Borel $\sigma$ algebra on $\mathbb{R}$, $\Lambda_1(.)$ is the asymptotic maxima distribution and $a$ is a positive constant. Hence, we can see that the maximum deviation between the true distribution of the maximum and the asymptotic distribution of the maxima over all the points, decreases with an increase in $L$ or $1/\beta$. (Note that, this is the rate of convergence at the point of maximum possible deviation over the entire support of the maximum distribution. We can expect faster convergence over some subsets of the support of the maxima distribution). Further, this result says that the rate of convergence is determined by the number of interferers $N$ and the number of clusters $\mu_i$ (for $i=1,\cdots,N$) in the interferers' fading distribution. The convergence rate decreases as the number of interferers increases or the number of clusters $\mu_i$ increases. Thus, the distributions of $\mathlarger{\gamma^{L}_{max}}$ for interferers with fading environments having $\mu_i=1$ (Rayleigh, Rician or shadowed Rician) converge faster to the asymptotic distribution of the maximum, given by (\ref{asymp_cdf}) than those having $\mu_i>1$ (Nakagami-m, $\kappa$-$\mu$ or $\eta$-$\mu$). Also, note that the parameters $\kappa$, $\mu$ and $m$ of the source and the parameters $\kappa_i$ and $m_i$ (for $i=1,...,N$) of the interferers do not affect the convergence rate. 
\section{Analysis of extreme-value distribution} \label{analysis}
Now that we have derived the asymptotic distribution of the maximum of $L$ i.i.d. SIR RVs where the source and the interferers are assumed to experience i.n.i.d. $\kappa-\mu$ shadowed fading, in this section, we analyze the impact of fading parameters $\kappa$, $\mu$ and $m$ on the asymptotic distribution. For this, we give the following key lemma. 
\begin{lemma}\label{frechetlemma1}
Consider two Frechet RVs $P$ and $Q$ with parameters $\mathlarger{\{a_{L1}, \beta\}}$ and $\mathlarger{\{a_{L2}, \beta\}}$ respectively. $P$ is stochastically larger than $Q$ if
\begin{equation}
    \mathbb P(P < z) < \mathbb P(Q <z), \ \forall z >0. 
\end{equation}
In other words, $\mathlarger{P>_{st}Q}$ if
\begin{equation}
    exp\left(-\left( \frac{z}{a_{L1}}\right)^{-\beta}\right) < exp\left(-\left( \frac{z}{a_{L2}}\right)^{-\beta}\right).
\end{equation}
The above condition is achieved when $\mathlarger{a_{L1} \geq a_{L2}}$.
\end{lemma}

 According to Lemma \ref{frechetlemma1}, the variations in the asymptotic CDF of the maximum SIR is governed by the variations in $a_L$ where $\mathlarger{a_L=F_{\gamma}^{-1}(1-L^{-1})}$. Hence, the variation in the CDF of the maximum SIR with respect to the variations in the source fading environment can be studied by analyzing the variations in $a_L$. However, the relationship between various parameters and $a_L$ is highly non linear, and therefore comprehending these variations with respect to changes in the fading parameters is very difficult. One way to circumvent this problem is to use moment matching as in \cite{kumar2015approximate}, and approximate each of the $\kappa-\mu$ shadowed RV as a gamma RV. The $\kappa-\mu$ shadowed RV corresponding to the users fading coefficients with parameters ${(\kappa,\mu,m,\bar{x})}$ can be approximated with a gamma RV with shape parameter $\psi_1 = \frac{m\mu(1+\kappa)^2}{m+\mu\kappa^2+2m\kappa}$ and scale parameter $\psi_2=\frac{\bar{x}}{\psi_1}$. Similarly, each of the $\kappa-\mu$ shadowed interferer can be first approximated as a gamma RV and their sum can be further approximated by another gamma RV with parameters $(\phi_1,\phi_2)$ using \cite[Eqn. (4)]{kumar2015approximate}. Here, we have $\mathlarger{F_{\gamma}(z) = \mathbb{P}(\gamma \leq z)\approx \mathbb{P}\left( \frac{\boldsymbol{\Gamma}(\psi_1,\psi_2)}{\boldsymbol{\Gamma}(\phi_1,\phi_2)} \leq z \right)=\mathbb{P}\left(\frac{\boldsymbol{\Gamma}(\psi_1,1)}{\boldsymbol{\Gamma}(\phi_1,1)} \leq z \frac{\phi_2}{\psi_2} \right)}$, where $\mathlarger{\boldsymbol{\Gamma}(.,.)}$ represents a gamma distributed RV. This ratio of gamma RVs has a beta-prime CDF \cite{dubey1970compound} with parameters $\mathlarger{\psi_1}$ and $\mathlarger{\phi_1}$ evaluated at $\mathlarger{z\frac{\phi_2}{\psi_2}}$. Now, the analysis in \cite{secrecycapacity} can be used to make inferences about the approximate variation in $\mathlarger{F_{\gamma}(z)}$, with respect to the changes in $\kappa,\mu$ and $m$. Based on the analysis, we give the following observations.\\
 \textbf{\textit{Observation 1 } : Scale parameter of the Frechet distribution $\boldsymbol{a_L}$ increases with increase in $\boldsymbol{\mu}$ or $\boldsymbol{m}$.}\\
Observe that, an increase in $\mu$ or $m$ results in an increase in $\psi_1$. According to $I4$ in Section III of \cite{secrecycapacity}, with an increase in $\psi_1$ along with a proportionate increase in $\bar{x}$, we can observe a decrease in $F_{\gamma}(z)$. Since CDF is an monotonically increasing function, to obtain the same CDF value of $1-\frac{1}{L}$ even after an increase in $\mu$ or $m$, the CDF evaluation point, which in our case is $a_L$, has to increase.\\
\textbf{\textit{Observation 2 } : Scale parameter of the Frechet distribution $\boldsymbol{a_L}$ increases with increase in $\boldsymbol{\kappa}$ if $\boldsymbol{m -\mu \geq 0 }$ and decreases otherwise .} \\
The derivative of $\psi_1$ with respect to $\kappa$ is given by $\mathlarger{\frac{2\kappa(1+\kappa)m\mu(m-\mu)}{\left(m+2\kappa m+\kappa^2\mu \right)^2}}$. This shows that $\psi_1$ increases with an increase in $\kappa$ if $m-\mu > 0$ and decreases otherwise. This in turn implies that the scale parameter $\mathlarger{F_{\gamma}^{-1}(1-L^{-1})}$ increases with an increase in $\kappa$, if $m-\mu > 0$ and decreases otherwise. Hence, following the same reasoning given in \textit{Observation 2}, we can infer that an increase in $\kappa$ increases $a_L$, if $m-\mu>0$, owing to the increase in $\psi_1$. Similarly, an increase in $\kappa$ results in an decrease in $a_L$, if $m-\mu<0$.
\par Thus \textit{Observation 1}, \textit{Observation 2} along with Lemma \ref{frechetlemma1} gives inferences on the variation of the asymptotic maximum distribution with respect to the changes in the source's fading environment. Further, Table I in \cite{pozas_shadowed} summarizes the relation between $\kappa-\mu$ shadowed fading model and many common fading models like Rayleigh, Rician, Nakagami, etc. Using these results, we can analyze the variations in the maximum SIR for any specific fading environment as well. 

\subsection{KL divergence between asymptotic maximum distribution and the true maximum distribution}
To get a quantitative idea of how the convergence of the true distribution of the maximum to the asymptotic distribution of the maximum varies for different values of $L$ and $\beta$, we compute the empirical KL divergence between the maximum SIR samples and the samples from the corresponding Frechet distribution\footnote{Since the exact CDF of the maximum SIR has a complicated structure, it is mathematically intractable to derive an expression for the KL divergence and hence we calculate the empirical KL divergence}. To calculate the empirical KL-divergence, we use the method discussed in \cite{wang2005divergence}. 
Let $\{X_1,\cdots,X_n \}$ and $\{Y_1,\cdots,Y_n \}$ be i.i.d samples from the distributions $P$ and $Q$ respectively.
    Now, we compute their histograms over the complete range of samples divided into equispaced bins. According to the Freedman-Diaconis rule\footnote{$Number\ of\ bins = \frac{Max(\{X_i\})-Min(\{X_i\})}{2 \times IQR \times n^{-1/3}}$ where $IQR$ is the interquartile range.} \cite{freedman1981histogram}, the number of bins is computed for both set of samples and the maximum of the two, given by $W$, is chosen to compute the histograms. If $u_i$ and $v_i$ represent the number of samples in the $i^{th}$ bin of histograms of $P$ and $Q$ respectively, then the corresponding empirical KL divergence is computed as \cite{wang2005divergence} \begin{equation}
    D_{KL}(P||Q) \approx \sum\limits_{i=1}^W \frac{u_i}{n}\log\left(\frac{u_i}{v_i} \right). 
\end{equation} The following tables give the empirical KL divergence of the asymptotic distribution of the maximum from the true distribution of the maximum for different values of $N$, $L$ and $\beta$.  The number of samples in each case is $n=10^6$. Tables \ref{table_kl} (a)-(c) gives the KL divergence in Rayleigh fading scenario for different number of interferers. The smaller the KL divergence, the closer are two distributions. We can see that the KL divergence decreases as $L$ increases for all the cases as expected. Similarly, for the same value of $L$, we can see that the KL divergence increases with $N$. This observation is also in agreement with the rate of convergence results derived which says that the rate decreases with an increase in $\beta$. For the case of Rayleigh fading, we have $\beta=N$.
\begin{table}[H]
	\centering
\begin{tabular}{|c|p{2cm}|p{2cm}|p{2cm}||p{2cm}|p{2cm}|p{2cm}|}  
\hline                               
L & (a) KL divergence for N=1 & (b) KL divergence for N=2& (c) KL divergence for N=3 & (d) KL divergence for $\beta=2$ & (e) KL divergence for $\beta=3$ & (f) KL divergence for $\beta=4$\\    
\hline                               
                           
20 & 3.056835e-04 & 6.917400e-02 & 1.866447e-01 &8.416245e-02 & 8.327871e-02 & 4.312127e-01\\             
\hline                               
                              
40 & 2.401260e-04 & 3.431374e-02 & 1.365351e-01 &  3.17051e-02 & 2.096145e-02 & 4.113886e-01\\             
\hline                               
                               
60 & 1.740811e-04 & 2.289564e-02 & 9.692215e-02 &  2.654761e-02 & 1.222194e-02 & 3.745029e-01\\             
\hline                               
                               
80 & 1.173821e-04 & 1.725099e-02 & 8.344516e-02 &  1.877953e-02 & 1.163133e-02 & 3.373044e-01\\             
\hline                               
                              
100 & 9.642496e-05 & 1.346993e-02 & 6.688599e-02 & 1.041195e-02 & 1.120306e-02 & 2.796802e-01\\           
\hline                               
\end{tabular}                               
\caption{Empirical KL divergence values for Rayleigh and $\kappa-\mu$ shadowed fading}             
\label{table_kl}       
	\end{table}
Similarly, Tables \ref{table_kl} (d)-(f) give the empirical KL divergence for the case of $\kappa-\mu$ shadowed fading for different values of $\beta$. Here too, we can see that the KL divergence increases with a decrease in $L$ or an increase in $\beta$. From the above tables, it is clear that for $L>60$, the KL divergence values are small and hence our asymptotic distribution very well approximates the true distribution of the maximum, for maximum taken over sequences of length greater than 60.\footnote{Note that, similar behaviour of KL divergence is observed for different values of $\kappa$, $\mu$ and m. However, due to space constraints we have only included the results for a subset of cases.} This along with the rate of convergence analysis reaffirms the claim that our asymptotic results can be reliably used for the performance analysis in all scenarios where we need the statistics of the maximum SIR.

\section{Applications and Simulations} \label{results}
Our results can be used in any application which involves the maximum SIR statistic. Here, we present some example applications.

\subsection{Asymptotic outage probability and ergodic rate of MSC system.}
\par The limiting distribution of the maximum SIR RV is useful for the asymptotic performance analysis of channel-aware packet scheduling systems \cite{song2006asymptotic}. Consider a time-slotted downlink channel shared among $L$ users. With MSC scheduling, the channel is assigned to the user with the maximum SIR in each time slot. The authors of \cite{song2006asymptotic} have analyzed the performance of such a system in a noise-limited scenario under Rayleigh fading channel using EVT. We can readily generalize these results to a $\kappa-\mu$ shadowed fading environment using the results derived in the previous section. 
    \subsubsection{Outage probability} 
The probability of outage of the user in each time slot of an MSC scheduling system, for a threshold $\gamma^{}_T$, in an interference limited scenario is given by 
\begin{equation}
    \mathbb{P}(\gamma^{L}_{max} \leq \gamma_T) = \prod\limits_{j=1}^L F_{\gamma_i}(z),
\end{equation}
where $\gamma_i$ is the SIR of the $i^{th}$ user.  When all the users experience identical fading, the above probability is the same as $\left(F_\gamma(z) \right)^L$ evaluated at $\gamma_T$. Note that, the evaluation of the true distribution of the maximum given by $\left(F_\gamma(z) \right)^L$ is not computationally tractable in a $\kappa-\mu$ shadowed fading environment. However, the asymptotic distribution function derived in the previous section can be used to compute this outage probability easily. In the following figures, we show the simulated and asymptotic CDF of maximum for a two interferer scenario.

 \begin{figure}[H]
	\centering
	\begin{minipage}[t]{0.5\textwidth}
\includegraphics[scale=0.45]{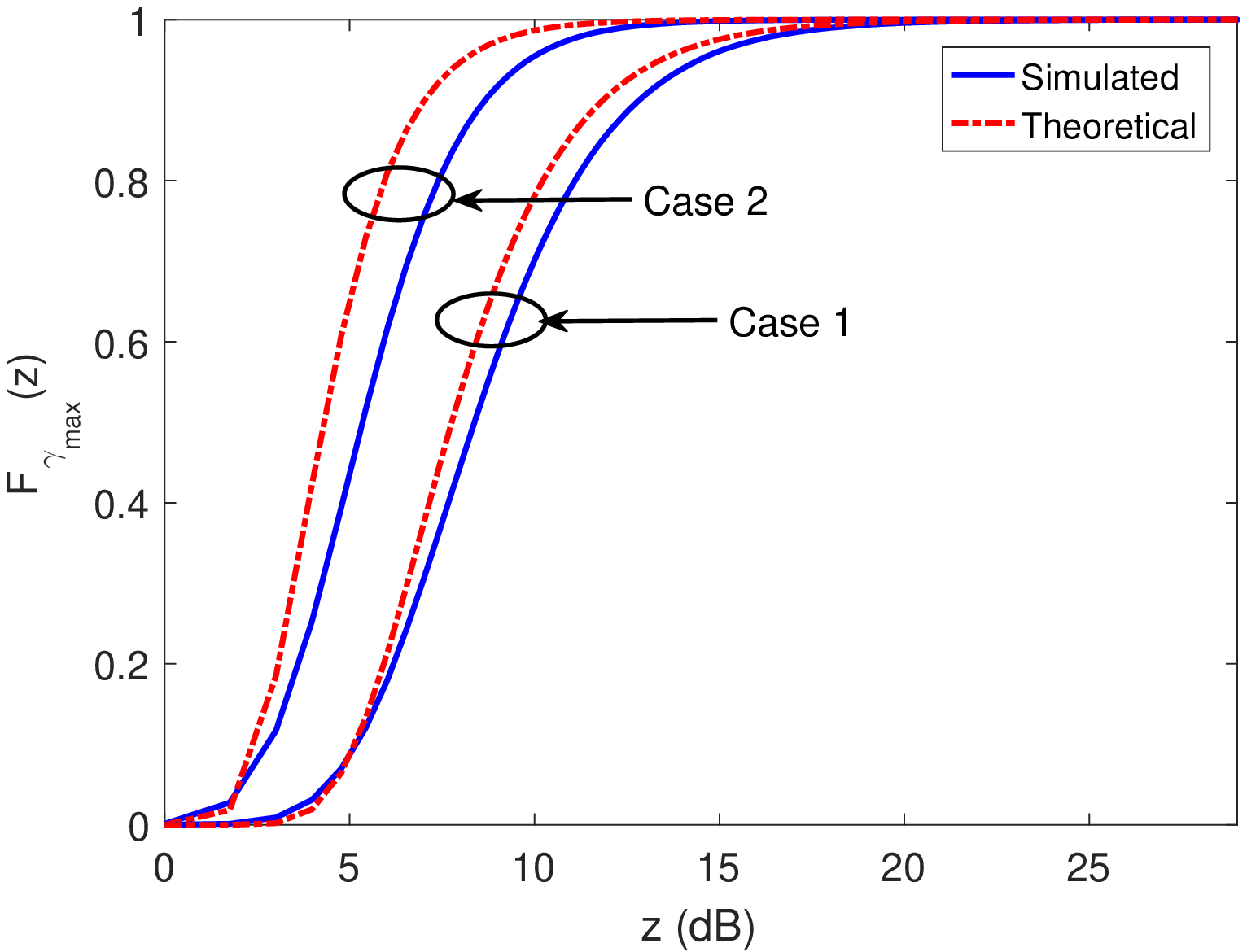}
		\caption{CDF of $\gamma^L_{max}$ for $\kappa-\mu$ shadowed fading with i.n.i.d. interferers, L=64. }
		\label{kus_inid1}
	\end{minipage}
	\begin{minipage}[t]{0.5\textwidth}
		\includegraphics[scale=0.45]{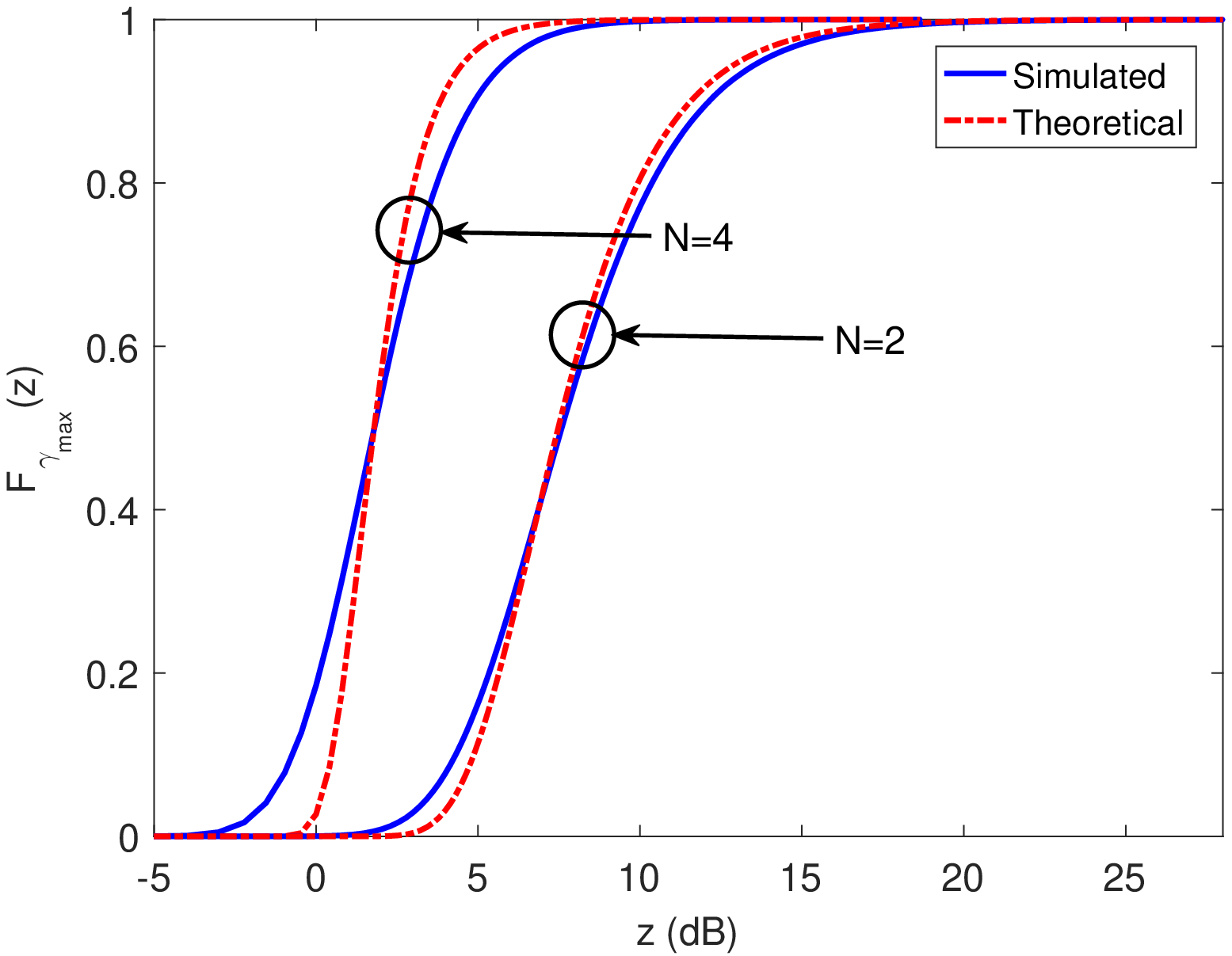}
		\caption{CDF of $\gamma^L_{max}$ for Rayleigh fading with i.i.d. interferers, L=32. }
		\label{rayl1}	
	\end{minipage}
\end{figure}
 Fig. \ref{kus_inid1}, shows the CDFs for the case of $\kappa-\mu$  shadowed fading with i.n.i.d interferers. Here, case 1 corresponds to the scenario where $\kappa=2,\mu=2,m=3$, $\kappa_i=\{ 2,2\}, \mu_i=\{2,1 \},m_i=\{3,2\}$ and case 2 corresponds to $\kappa=2,\mu=1,m=2$ and $\kappa_i=\{ 2,2\}, \mu_i=\{1,1 \},m_i=\{2,1\}$. Fig. \ref{rayl1} show the case of Rayleigh fading channel for different values of $N$ when $L=32$.Fig. \ref{kus_mu_vary1}-\ref{kus_k_vary1_mugm} validates Observation 1 and 2 in Section \ref{analysis}. An increase in $\mu$ or $m$ increases the scale parameter $a_L$. From Lemma \ref{frechetlemma1}, an increase in the scale parameter for a constant shape parameter results in a shift of the CDF to the right. This results in a lesser probability of outage for the same threshold. For clarity, we show the results only for $L=200$. 
 \begin{figure}[H]
	\centering
	\begin{minipage}[t]{0.5\textwidth}
	\includegraphics[scale=0.45]{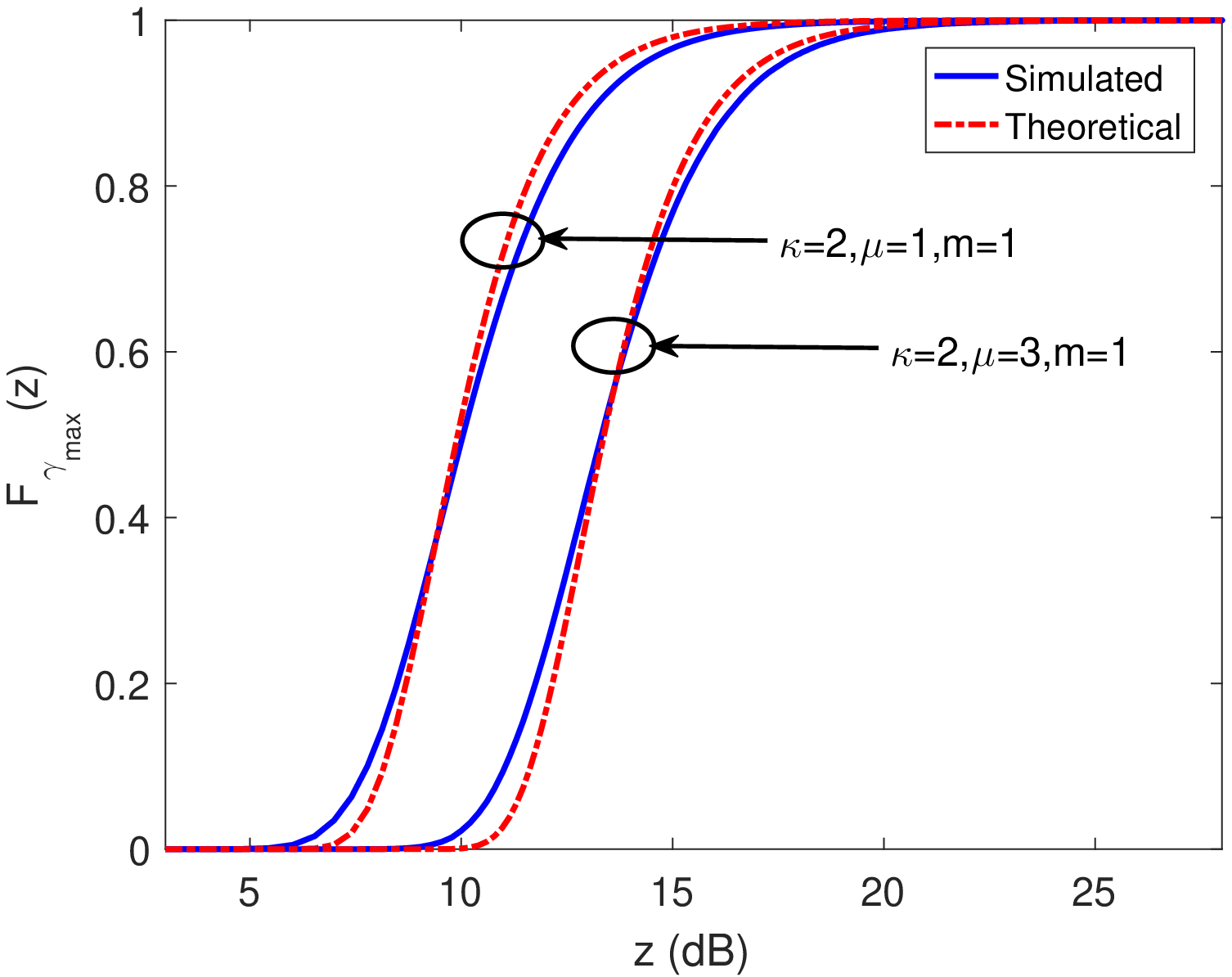}
		\caption{CDF for N=1, L=200, $\kappa_I=2$, $\mu_I=3$, $m=1$ . }
		\label{kus_mu_vary1}
	\end{minipage}
	\begin{minipage}[t]{0.45\textwidth}
		\includegraphics[scale=0.45]{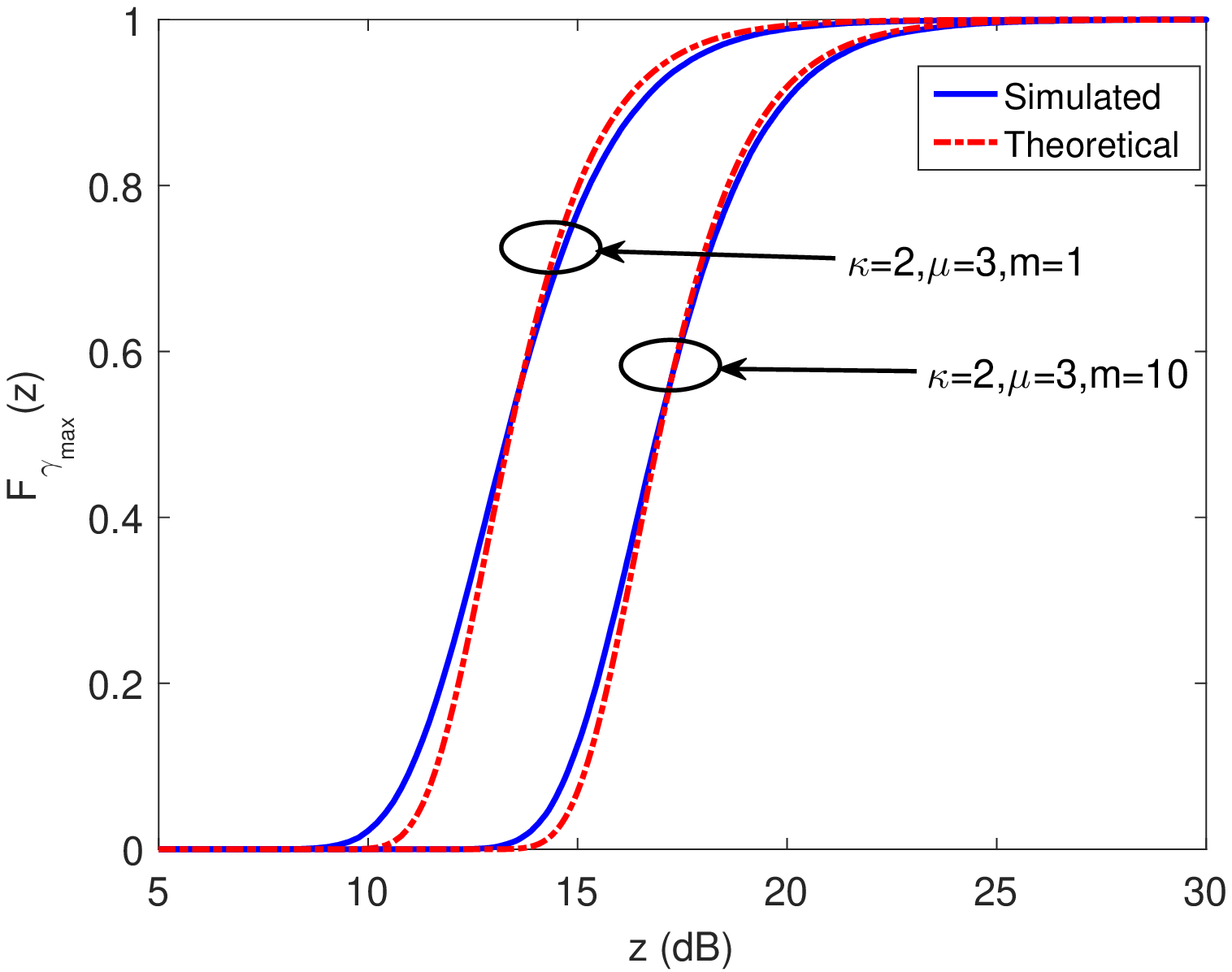}
		\caption{CDF for N=1, L=200, $\kappa_I=2$, $\mu_I=3$, $m=1$ . }
		\label{kus_m_vary1}		
	\end{minipage}
\end{figure}
 \begin{figure}[H]
	\centering
		\begin{minipage}[t]{0.45\textwidth}
	\includegraphics[scale=0.45]{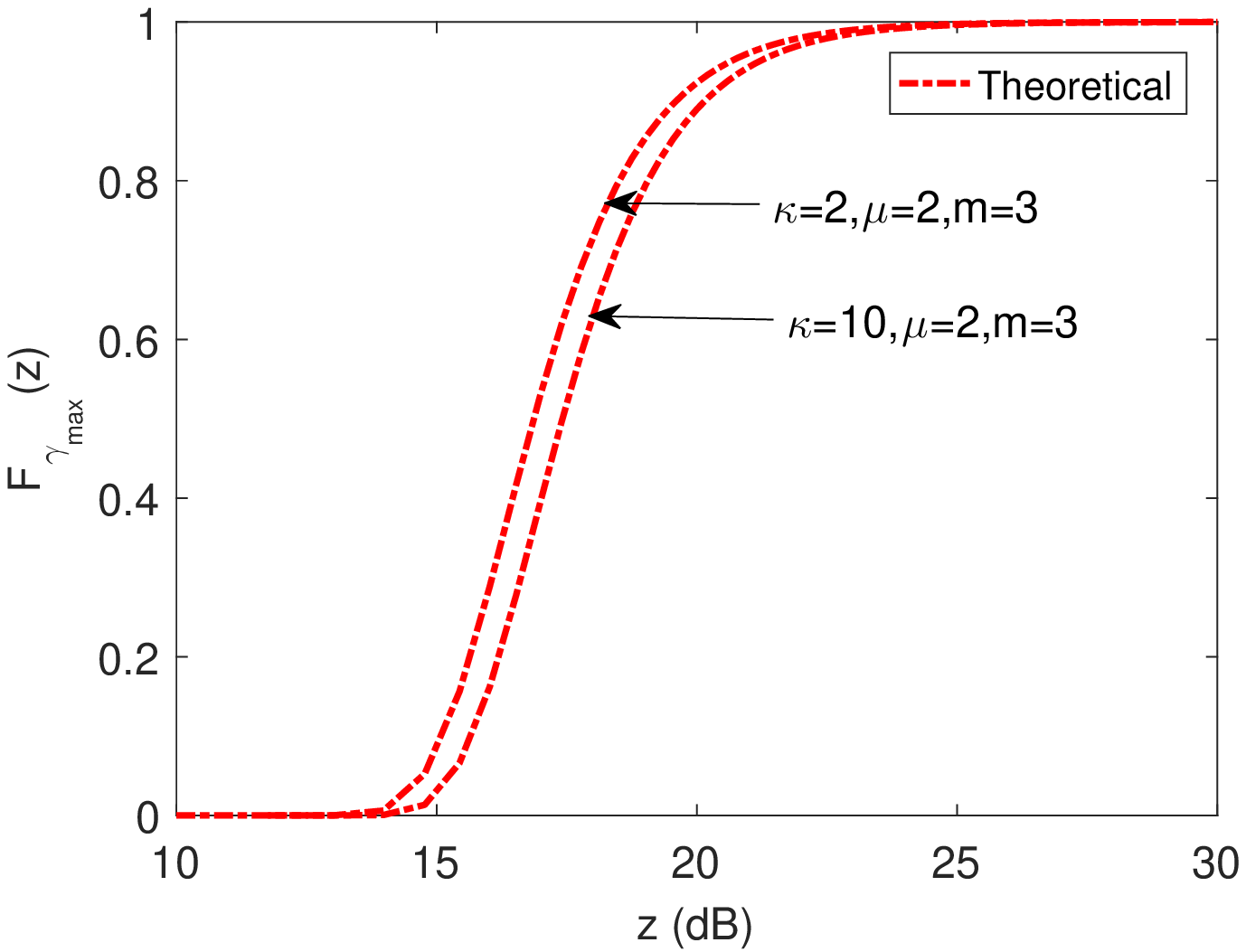}
		\caption{CDF for N=1, L=200, $\kappa_I=2$, $\mu_I=3$, $m=1$ .}
		\label{kus_k_vary1_mgmu}
	\end{minipage}
	\begin{minipage}[t]{0.45\textwidth}
		\includegraphics[scale=0.45]{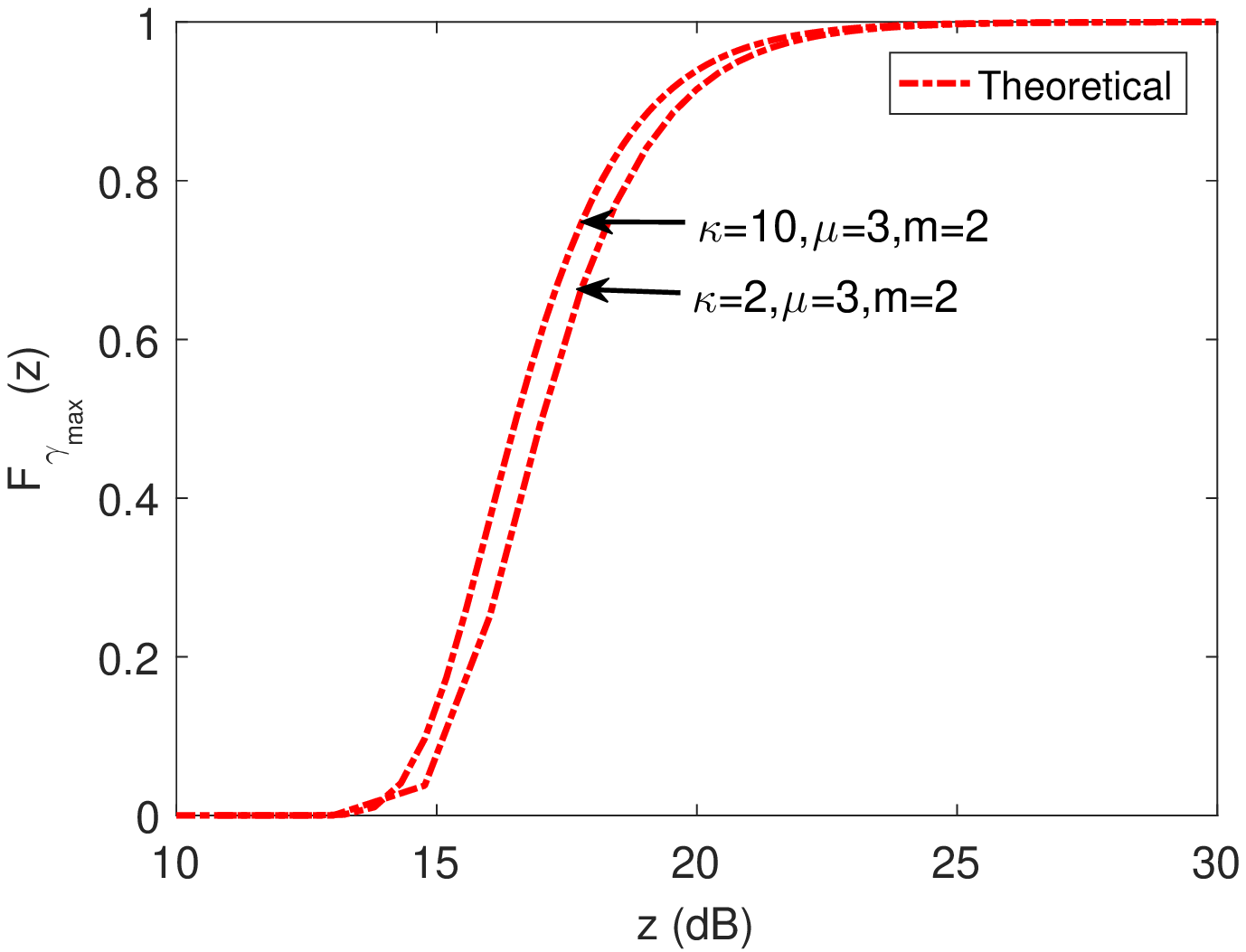}
		\caption{CDF for N=1, L=200, $\kappa_I=2$, $\mu_I=3$, $m=1$ .}
		\label{kus_k_vary1_mugm}		
	\end{minipage}
\end{figure}

Fig. \ref{kus_k_vary1_mgmu} and \ref{kus_k_vary1_mugm} show the variation of CDF with respect to variation in $\kappa$. From \textit{Observation 2} in Section \ref{analysis}, we know that, with an increase in $\kappa$ the scale parameter $a_L$ decreases, if $m-\mu$ is positive. This would result in a decrease in the outage probability. Fig. \ref{kus_k_vary1_mgmu} shows such a scenario and the result agrees with the expected observation. Fig. \ref{kus_k_vary1_mugm} corresponds to a case where $\mu > m$. In this case, it can be observed that an increase in $\kappa$ results in an increase in the outage probability. The change in outage probability with the change in $\kappa$ is not very large. Hence, for clarity, in Figs. \ref{kus_k_vary1_mgmu} and \ref{kus_k_vary1_mugm}, we have given only the theoretical Frechet distribution curves.

Fig. \ref{moment1} compares the simulated and theoretical values of the first moment of $\gamma^{L}_{max}$ for different values of $L$ and $N$. As discussed in Section \ref{sectionB}, the first moment of the asymptotic distribution converge to the first moment of the original distribution of $\gamma^{L}_{max}$. From the results, it is clear that the simulated and theoretical values of expectation get closer as $L$ increases and this convergence is faster for smaller values of $N$. 
\subsubsection{Ergodic rate}
\textbf{The asymptotic ergodic rate of the user in each time slot of the MSC scheduling system is given by} 
\begin{equation}
    R^L_{max} = \mathbb{E}\left[ \log_2(1+\gamma^L_{max})\right]. 
    \label{rate11}
\end{equation}
where $\gamma^L_{max} =max\{\gamma_1,\cdots,\gamma_L\}$. Recall that $\gamma^L_{max}$ converges in distribution to a Frechet RV ${\gamma}_{max}$, i.e., $\gamma^L_{max} \xrightarrow[\text{}]{\text{D}}{\gamma}_{max}$. We still have to prove that $\lim\limits_{L \to \infty} \mathbb{E}[R^L_{max}] = \mathbb{E}[{R}_{max}] $ where $R_{max}=log_2(1+\gamma_{max})$. To prove this, we first utilize continuous mapping theorem, which is given as follows \cite{billingsley2013convergence}: 
\begin{theorem} \label{cts_map}
	Let $\{X_n\}_{n=1}^\infty$ be a sequence of random variables and $X$ another random variable, all taking values in the same metric space $\mathcal{X}$. Let $\mathcal{Y}$ be another metric space and $f:\mathcal{X} \to \mathcal{Y}$ a measurable function and $C_f := \{x : \ f \ is \ continuous \ at \ x\}$. Suppose that $X_n \xrightarrow[\text{}]{\text{D}}X$ and $\mathbb{P}(X \in C_f)=1$, then $f(X_n) \xrightarrow[\text{}]{\text{D}}f(X)$.
\end{theorem}
Let $R^L_{max}=log_2(1+\gamma^L_{max})$.  Since $f(x)=log_2(1+x)$ is a continuous function, using Theorem. \ref{cts_map}, $R^L_{max} \xrightarrow[\text{}]{\text{D}}{R}_{max}$. Finally, we use monotone convergence theorem, which is given below \cite{billingsley2008probability}.
\begin{theorem} \label{mct}
	Let $g_n \geq 0$ be a sequence of measurable functions such that $g_n(\omega) \to g(\omega) \ \forall \ \omega$ except maybe on a measure zero set and $g_n(\omega) \leq g_{n+1}(\omega), \ n \geq 1$. We then have
	\begin{equation}
		\lim\limits_{n \to \infty} \int g_n \ d\mu = \int g \ d\mu.
	\end{equation}
\end{theorem}
Here, we know that $\gamma^L_{max} \leq \gamma^{L+1}_{max}, \ \forall \ L$ and hence $\mathbb{P}(\gamma^L_{max} \leq l) \geq \mathbb{P}(\gamma^{L+1}_{max} \leq l)$. Thus, $1-F_{\gamma^L_{max}}(l) \leq 1-F_{\gamma^{L+1}_{max}}(l)$. Logarithm is a monotonic function and hence  $1-F_{R^L_{max}}(l) \leq 1-F_{R^{L+1}_{max}}(l)$. For a positive RV $X$, note that the expectation is given by 
\begin{equation}
	\mathbb{E}[X] = \int\limits_{0}^{\infty} \mathbb{P}(X>x) \ dx = \int\limits_{0}^{\infty} (1-F_X(x)) \ dx.
\end{equation} 
Thus, making use of Theorem. \ref{mct} we have $\lim\limits_{L \to \infty} \mathbb{E}[R^L_{max}] = \lim\limits_{L \to \infty} \int\limits_{0}^{\infty}  \mathbb{P}(R^L_{max}>l) \ dl= \int\limits_{0}^{\infty} \lim\limits_{L \to \infty} \mathbb{P}(R^L_{max}>l) \ dl = \mathbb{E}[{R}_{max}] $. Hence, we have the required result. Given that we have proved the convergence of moments of $\mathlarger{R^L_{max}}$ to the moments of $\mathlarger{{R}_{max}}$, for large $L$, the expectation in (\ref{rate11}) can now be evaluated using the CDF of Frechet RV given in (\ref{asymp_cdf}), instead of using the true CDF of $\mathlarger{\gamma^L_{max}}$, which is difficult to evaluate. The asymptotic ergodic rate is thus given by $R = \int\limits_0^{\infty} \log_2(1+z) \ {f}_{{\gamma}_{max}}(z) \ dz,$ where $\mathlarger{{f}_{{\gamma}_{max}}(z)}$ is the asymptotic PDF of the Frechet RV $\mathlarger{\gamma^L_{max}}$. Substituting the Frechet PDF, the previous expression can be rewritten as follows :
\begin{equation}
    R = \beta(a_L)^\beta \ \int\limits_{0}^\infty \log_2(1+z)z^{-\beta-1}e^{-\left(\frac{z}{a_L} \right)^{-\beta}}  \ dz. 
    \label{rate_theo_al}
\end{equation}
 \begin{figure}[H]
	\centering
	\begin{minipage}[t]{0.5\textwidth}
	\includegraphics[scale=0.5]{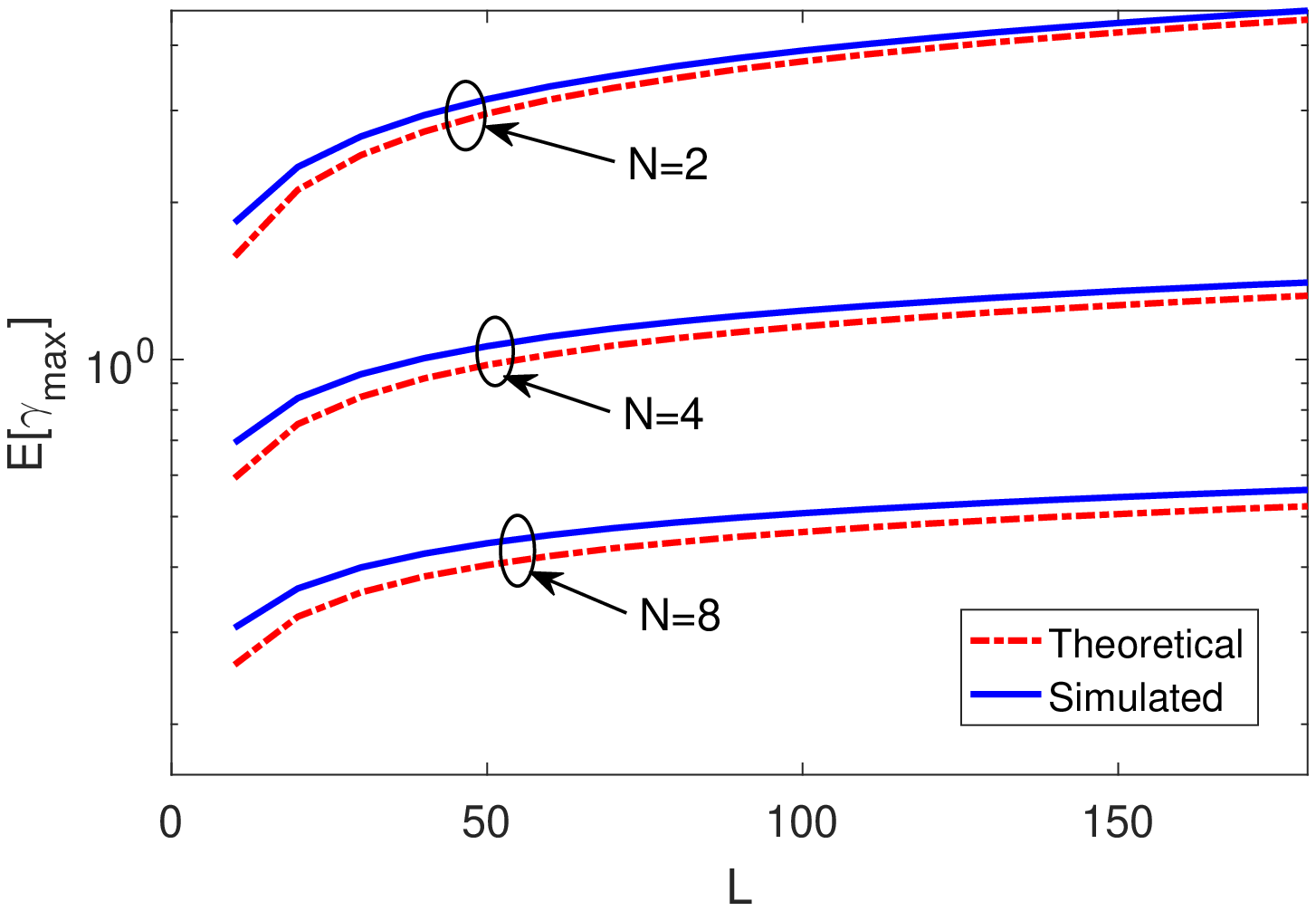}
		\caption{L vs $\mathbb{E}[\gamma^L_{max}]$ for $\kappa=2$, $\bar{\mu}=3$, $m=2$, $\kappa_I=2$, $\mu_I=2$, $m_I=3$. }
		\label{moment1}
	\end{minipage}
	\begin{minipage}[t]{0.5\textwidth}
		\includegraphics[scale=0.5]{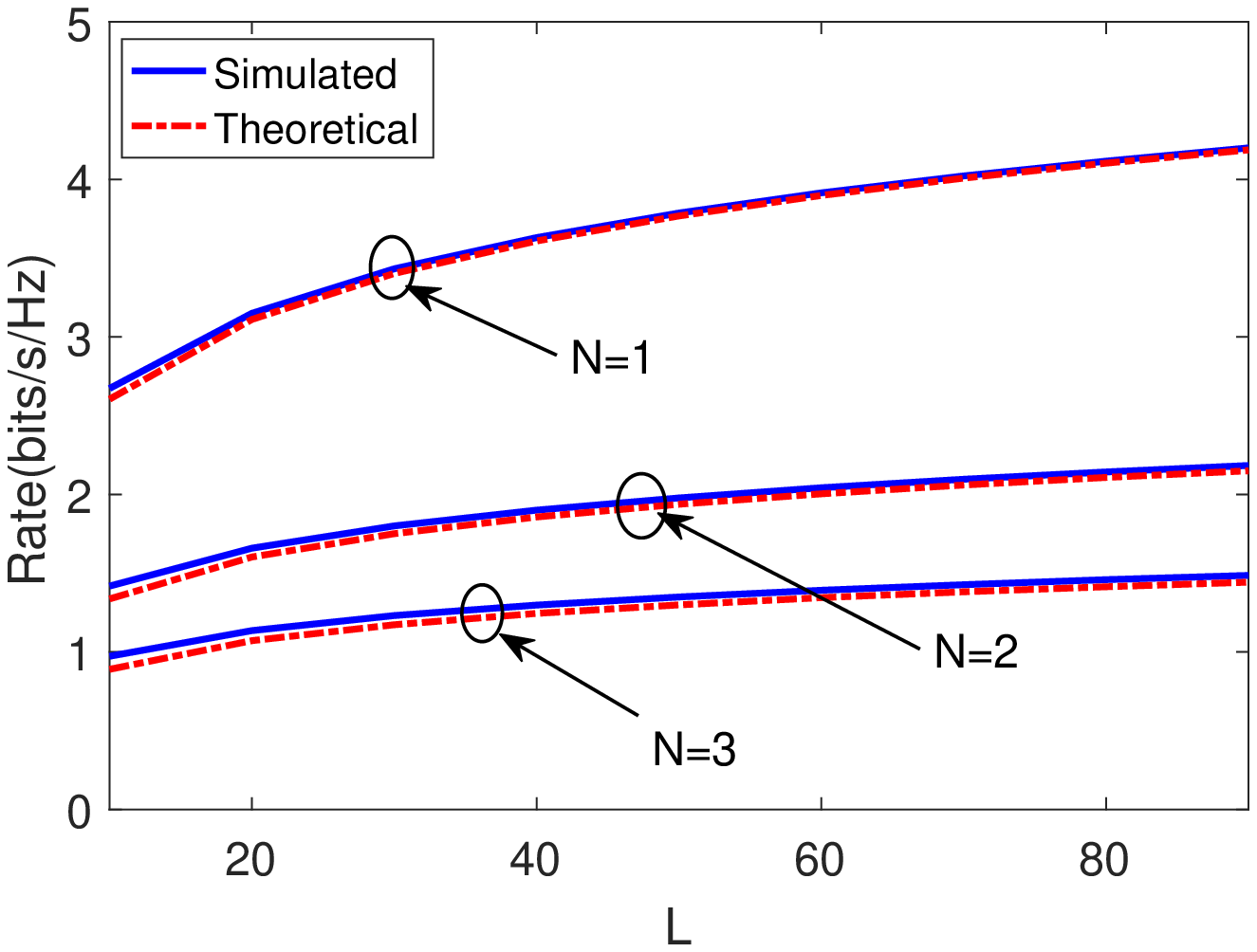}
		\caption{L vs ergodic rate for $\{\kappa=2,\mu=3,m=2\}$,$\{\kappa_I=2,\mu_I=2,m_I=3\}$. }
		\label{rate1}		
	\end{minipage}
\end{figure}
From (\ref{rate_theo_al}), we obtain the following observation : \\
\textbf{\textit{Observation 3 :}} \textbf{Ergodic rate increases with increase in $a_L$, for a constant shape parameter $\beta$.}\\
Note that both $\beta$ and $z$ are non-negative and hence $\mathlarger{\log_2(1+z)z^{-\beta-1}}$ will also be non-negative for all values of $\beta$ and $z$. Hence, with an increase in $a_L$, the rate increases. \textbf{\textit{Observation 1} and \textit{Observation 2} along with \textit{Observation 3} will facilitate obtaining inferences on the variations of the asymptotic data rate with respect to the $\kappa-\mu$ shadowed fading parameters. }

\par Fig. \ref{rate1} compares simulated values of the rate with the rate computed using (\ref{rate_theo_al}), for different values of $L$ and $N$. The results show that there is a good match between the simulated and theoretical values over a wide range of $L$ and $N$. As expected, the rate increases with the number of antennas at the receiver. However, with an increase in the number of interferers SIR signals of smaller magnitude are available at the receivers and hence the rate reduces significantly. 
 \begin{figure}[H]
	\centering
		\begin{minipage}[t]{0.5\textwidth}
	\includegraphics[scale=0.5]{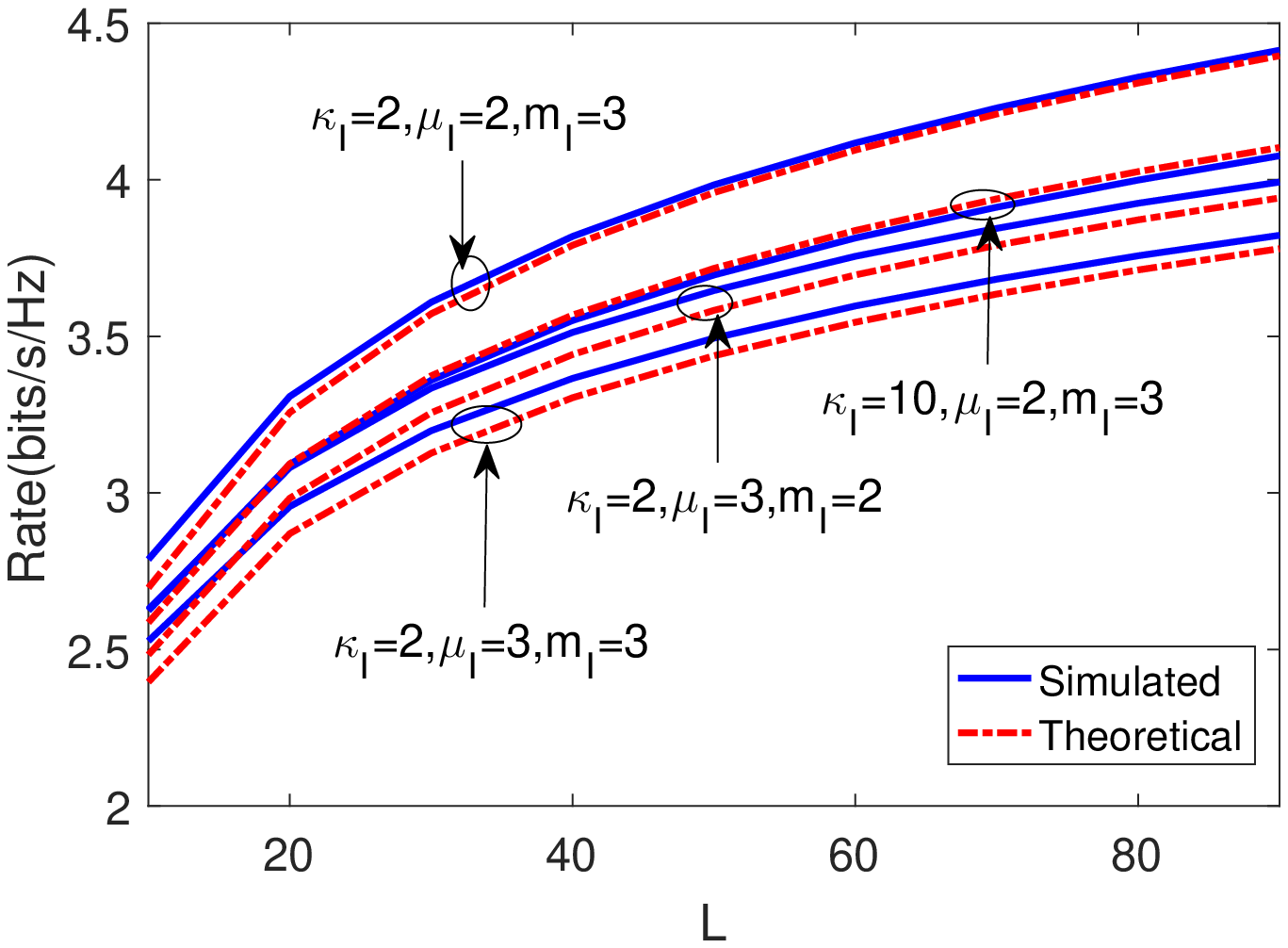}
		\caption{L vs ergodic rate for $\{\kappa=2,\mu=1,m=1\}$, N =1.}
		\label{rate2}
	\end{minipage}
	\begin{minipage}[t]{0.5\textwidth}
		\includegraphics[scale=0.5]{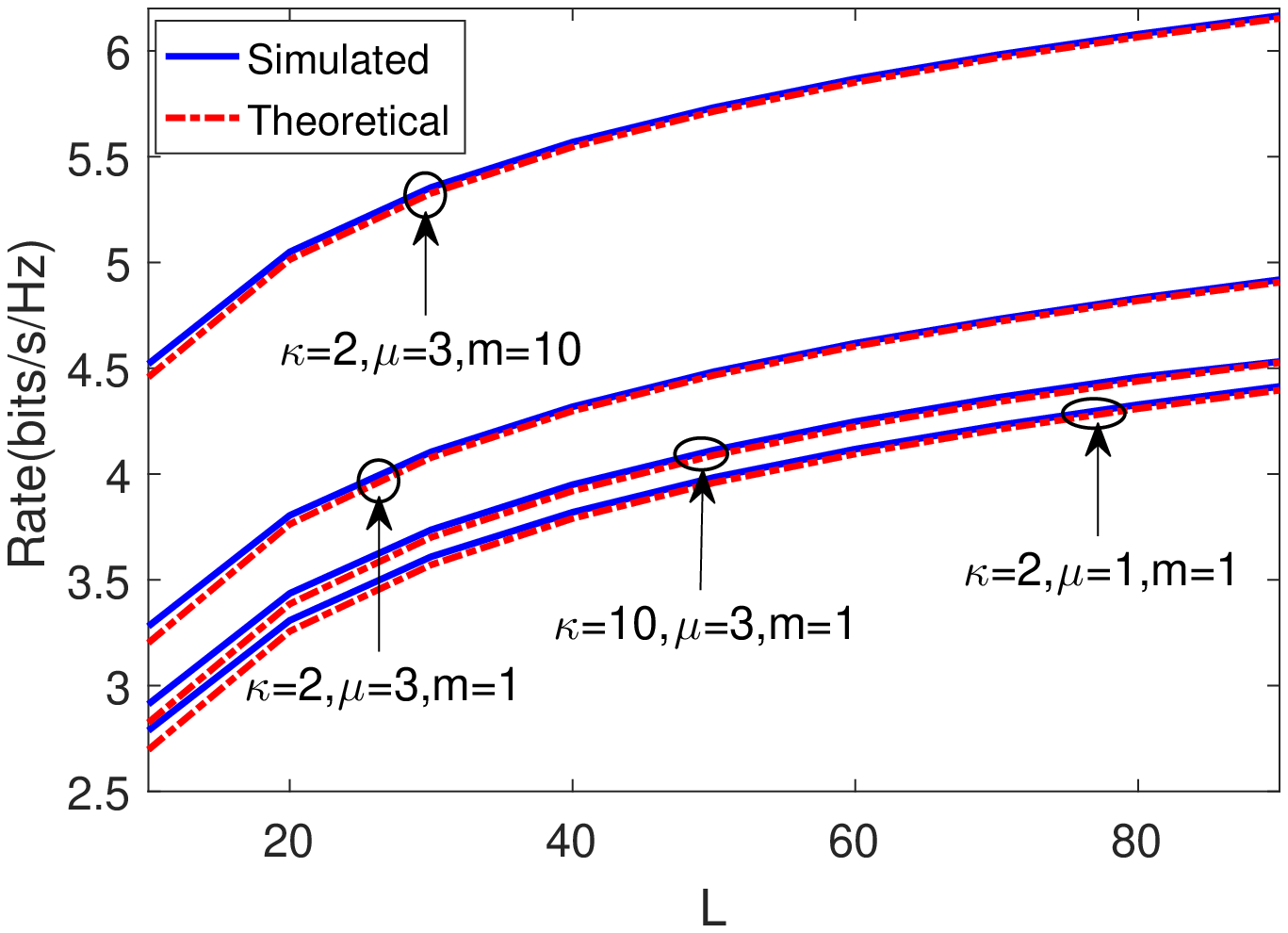}
		\caption{L vs ergodic rate for $\{\kappa_I=2,\mu_I=2,m_I=3\}$, N =1.}
		\label{rate3}		
	\end{minipage}
\end{figure}
Fig. \ref{rate2} further shows variation in the rate for changes in the interferer fading parameters (here, interferers are assumed to be i.i.d.). Increase in $\mu_I,m_I$ results in stronger interferers and hence results in a decrease in the rate. Further, Fig. \ref{rate3} shows the variation in rate for different values of source fading parameters. As discussed previously for the case of outage probability, increase in $m$, $\mu$ results in better coverage conditions and hence higher rates. Also, similar to the case of outage probability, the variation of the rate with respect to $\kappa$ depends on the sign of $\mu-m$. 

\subsection{Upper bound for the rate in FAS (Full Antenna Selection) architecture in massive MIMO}
In massive MIMO scenarios, there is often a restriction to the number of RF chains available for processing. Dedicated RF chains for each antenna in the massive MIMO antenna array is a cost prohibitive and power hungry design \cite{gao2016antenna}. Hence, a very common practice is to choose a subset of antennas from the array for further processing. This is known as the full array switching (FAS) architecture. Very recently, the authors of \cite{gao2016antenna,gao2018massive} discussed the analysis of FAS systems in a noise-limited scenario using EVT, when the channel experiences Rayleigh fading. In this subsection, we derive upper bounds on the rate for FAS architecture in an interference-limited scenario for $\kappa-\mu$ shadowed fading channels. Note that our results will generalize the FAS results of \cite{gao2018massive}. \\ Consider a multi-antenna receiver with $L$ antennas, out of which $L_s$ antennas are selected for further processing. Upper bound on the rate in a FAS scenario would correspond to the condition where the first $L_s$ antennas in terms of SIR are selected. The corresponding bound is given by \cite{gao2018massive} :
\begin{equation}
    R_{ub,fas} = \mathbb
    {E} \left[\sum\limits_{l=1}^{L_s}log_2(1+\gamma^L_{(l)}) \right],
    \label{rate_ub_fas}
\end{equation}
where $\{\gamma^L_{(l)} \}_{l=1,2,\cdots,L}$ are the ordered SIR RVs; ie. $\gamma^L_{(1)}>\gamma^L_{(2)}>\cdots>\gamma^L_{(L)}>0$. The convergence of the above moment to the first moment of the limiting distribution is guaranteed by extending the same claims used in the previous section to prove convergence of rate in MSC scheduling systems. Now, to compute the ergodic rate as in (\ref{rate_ub_fas}), we need the joint distribution of the first $L_s$ SIR RVs. The exact expression for the joint pdf of $L_s$ ordered i.i.d. random variables $\{x_{(1)}\geq x_{(2)}\geq \cdots\geq x_{(L_s)} \}$ is given by
\begin{equation}
    f_{x_{(1)},x_{(2)},\cdots,x_{(L_s)}}(x_{(1)},x_{(2)},\cdots,x_{(L_s)}) = L_s! \prod\limits_{l=1}^{L_s}f(x_{(l)}),
    \label{joint_true}
\end{equation} where $f(x)$ is the pdf of each of the RV $x_{(i)}; i=1,\cdots,L.$
The evaluation of (\ref{rate_ub_fas}) using the true joint distribution as given in (\ref{joint_true}) will result in a very complex expression due to the complicated nature of the $L_s$ product terms and the $L_s$ fold integration. In fact, even the evaluation of a single term of the product using (\ref{pdf_exp}) will take close to an hour in Mathematica. However, we make use of the following result from EVT to derive the asymptotic joint distribution of the first $L_s$ RVs.
\begin{lemma}
Given a sequence of i.i.d RVs $X_1,X_2,\cdots,X_L$ with a common CDF $F(x)$ that belongs to the MDA of one of the three EVD $G(x)$ such that $\frac{max(X_1,X_2,\cdots,X_L)-b_L}{a_L} \xrightarrow{D}G(x)$. Suppose that $X_{(1)}>X_{(2)}>\cdots>X_{(L)}$ is the ordered sequence of $X_1,X_2,\cdots,X_L$ then the $L_s$ dimensional vector $\left( \frac{X_{(1)}-b_L}{a_L}, \frac{X_{(2)}-b_L}{a_L},\cdots, \frac{X_{(L_s)}-b_L}{a_L}\right)$ has the following asymptotic joint distribution :
\begin{equation}
    g_{1,2,\cdots,L_s}(x_{(1)}, x_{(2)}, \cdots, x_{(L_s)})=G(x_{(L_s)})\prod\limits_{l=1}^{L_s}\frac{g(x_{(l)})}{G(x_{(l)})},
    \label{joint_dist_ordered_RV}
\end{equation}
with $x_{(1)}>x_{(2)}>\cdots>x_{(L_s)}$ and $g(x)$ is the pdf of $max(X_1,X_2,\cdots,X_L)$.
\end{lemma}
\begin{proof}
	Please refer to page 219 in \cite{arnold1992first} for the proof.
\end{proof}
 \begin{figure}[H]
	\centering
		\includegraphics[scale=0.45]{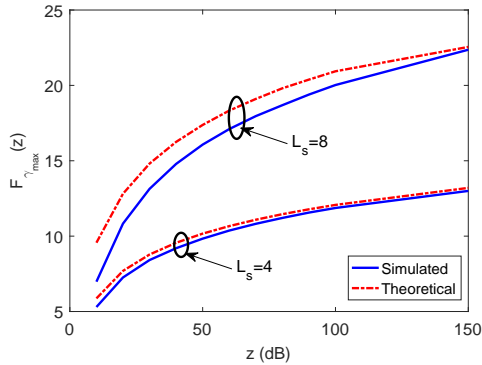}
		\caption{Asymptotic upper bound of rate and simulated rate of FAS systems in Rayleigh fading.  }
		\label{fas_rayl}
\end{figure}

The corresponding joint distribution in our case is as follows :
\begin{align}
        f_{x_{(1)},x_{(2)},\cdots,x_{(L_s)}}(x_{(1)},x_{(2)},\cdots,x_{(L_s)}) =  \left(\beta (a_L)^{\beta+1} \right)^{L_s} \ exp\left(-\left( \frac{x_{(L_s)}}{a_L}\right)^{-\beta} \right) \prod\limits_{l=1}^{L_s} \left(x_{(l)} \right)^{-1-\beta}.
        \label{joint_frechet}
\end{align}

 \begin{figure}[H]
	\centering
	\begin{minipage}[t]{0.45\textwidth}
		\includegraphics[width=\textwidth]{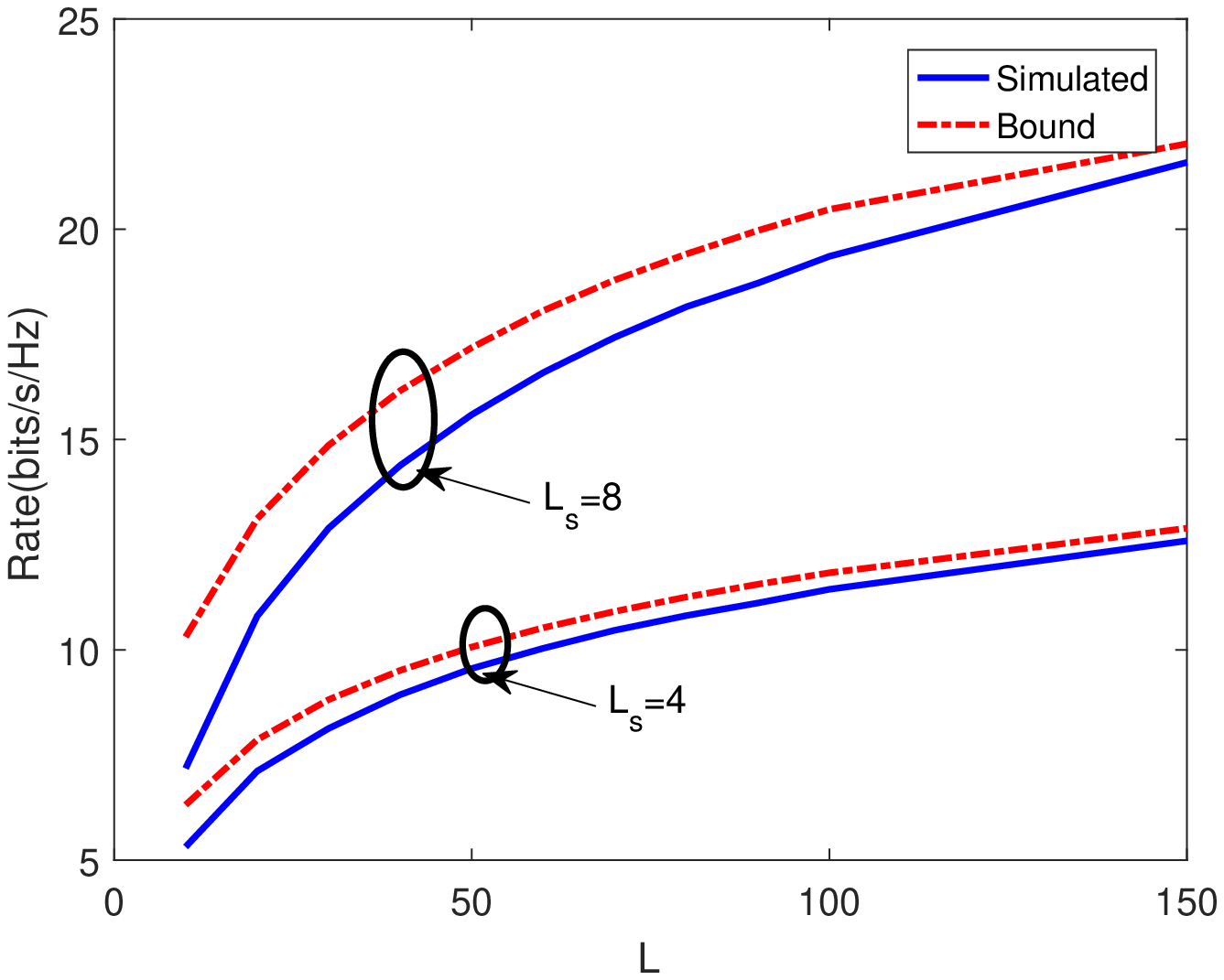}
		\caption{Asymptotic upper bound of rate and simulated rate of FAS systems in $\kappa-\mu$ shadowed fading with i.i.d. interferers. }
		\label{fas_kus_iid}
	\end{minipage}
	\hfill
	\begin{minipage}[t]{0.45\textwidth}
			\includegraphics[width=\textwidth]{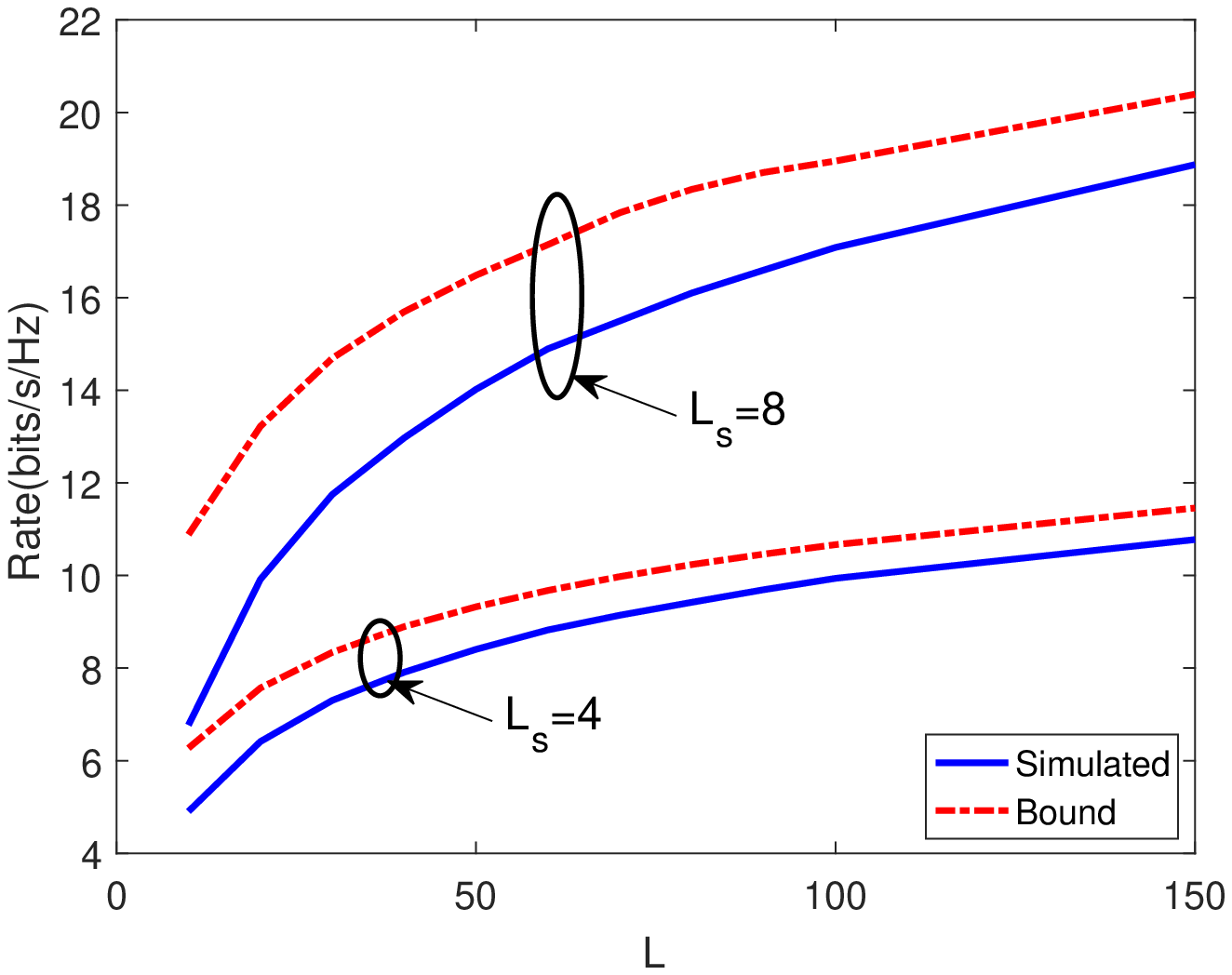}
		\caption{Asymptotic upper bound of rate and simulated rate of FAS systems in $\kappa-\mu$ shadowed fading with i.i.d. interferers.  }
		\label{fas_kus_inid}
	\end{minipage}
\end{figure}
Now, using the expression for the joint distribution of ordered RVs in (\ref{joint_frechet}) we can compute the upper bound for the rate easily. We show typical simulation results to validate the bounds for FAS systems. Here, the best 4 or 8 antennas are selected for further processing in each case. We can see that the bound is tight for the case of $L_s=4$ and the bound gets looser for larger values of $L_s$. Here, the asymptotic statistics of the ordered RVs allows easy evaluation of the upper bound on the rate.
Fig. \ref{fas_rayl} shows the asymptotic bound and the simulated rate with different antenna lengths, for the case of FAS architectures in Rayleigh fading environment. Further, Fig. \ref{fas_kus_iid} and Fig. \ref{fas_kus_inid} shows the simulated and asymptotic bound for the rate in FAS system in a $\kappa-\mu$ shadowed fading environment. The fading parameters chosen for the simulations in Fig. 12 and Fig. 13 are $\kappa=2$, $\mu=3$, $m=1$, $\kappa_I=2$, $\mu_I=1$, $m_I=1$ and  $\kappa=2$, $\mu=3$, $m=1$, $\kappa_1=2$, $\mu_1=1$, $m_1=1$ and $\kappa_2=2$, $\mu_2=2$, $m_2=1$ respectively. 

\subsection{Other applications:}
Besides the performance analysis of MSC scheduling and FAS systems, the results derived in our work can also be used in other applications. One can derive the asymptotic outage probability and ergodic rate of the maximum SIR antenna (i.e., selection combining) in a massive MIMO scenario. These metrics describe the performance of the best antenna from the antenna array under a given channel condition and hence are useful for system characterization and resource allocation, especially when limited number of high-resolution RF chains are available at the receiver \cite{fan2015uplink,srinivasan2019analysis}. Also, note that the asymptotic distribution of the $k$th maximum SIR RV can be derived easily from the distribution of the maximum RV using \cite[Eqn. (8.4.2)]{arnold1992first}. This can be used for the performance analysis of a general selection-diversity (SD) scheme as in \cite{al2018asymptotic}. One can also determine the joint distribution of the largest and second largest SIR RV and determine performance bounds for generalized selection combining (GSC) schemes in massive MIMO receivers \cite{kong2000snr,annamalai2001error,chen2004unified}. Also, our asymptotic distribution of the maximum SIR can be used for the analysis of receive antenna selection schemes in spectrum sharing systems \cite{hanif2018antenna}. 

\section{Conclusions} \label{conclusion}
We considered $L$ wireless links in the presence of $N$ co-channel interferers experiencing non-identical $\kappa-\mu$ shadowed fading conditions with $\{\kappa,\mu,m,\bar{x} \}$ and $\{\kappa_i,\mu_i,m_i,\bar{y_i};i=1,\cdots,N \}$ being the fading parameters of the source and interferers respectively. Let, $\gamma=\gamma_j, \ \forall \ j \ \in \ \{1,\cdots,L\}$ represent the SIR available at each user/antenna and $\gamma^L_{max}$ is the maximum over all $\gamma_j$. Following are our key results and observations :
\begin{itemize}
    \item The asymptotic distribution of $\gamma^L_{max}$ is a Frechet distribution with scale parameter $a_L=F_{\gamma}^{-1}(1-L^{-1})$ and shape parameter $\beta=\sum\limits_{i=1}^N\mu_i$.
    \item The $\nu^{th}$ moment of $\gamma^L_{max}$ converges to the $\nu^{th}$ moment of the corresponding Frechet distribution for all $\nu \ < \ \sum\limits_{i=1}^{N}\mu_i$. 
    \item 	The rate of convergence of $F_{\gamma^L_{max}}(z)$ to the Frechet distribution is $O\left( L^{-\left(\sum\limits_{i=1}^{N}\mu_{i} \right)^{-1}} + L^{-1} \right)$. From this result, we can observe that the maximum deviation between the true distribution of the maximum and the asymptotic distribution of the maximum over all the points decreases with an increase in $L$. Also, in order to demonstrate the practical validity of the work, we studied the empirical KL divergence between the exact distribution of the maximum and the corresponding asymptotic distribution. The KL divergence results indicated the closeness between the asymptotic results and the exact results, even for finite $L$.
    \item Using results from stochastic ordering, the variation in the behaviour of the asymptotic outage probability and asymptotic ergodic rate with respect to variations in the source fading environment is studied.
    \item We analyze the utility of the derived asymptotic results in the following applications : 
	    \begin{enumerate}[label=\roman*]
	        \item Analysis of asymptotic outage probability and asymptotic ergodic rate of the user in each time slot of an MSC system.
	        \item Derivation of the asymptotic upper bound on the rate in FAS architectures for antenna selection
	    \end{enumerate}
\end{itemize}
	        
Further, simulations are provided to validate the above results and to confirm the utility of our results. As the $\kappa-\mu$ shadowed fading model is a generalized fading model encompassing most of the general fading scenarios as special cases, our results can be used in a number of problem scenarios involving maxima statistics.

\begin{appendices}
\section{Proof for Theorem \ref{mda_frechet}}\label{proof_mda}
	$\mathlarger{F_{\gamma}(z)}$ belongs to the $MDA$ of the Frechet distribution, if it satisfies (\ref{codn_frec1}). The CDF of the SIR RV $\gamma$ is given by $F_{\gamma}(z)=\mathbb{P}(\gamma\leq z),$, where $\mathlarger{\mathbb{P}(.)}$  represents the probability of an event. Note that, this is equivalent to the expression for outage probability with a threshold $z$. The expression for outage probability in a $\kappa-\mu$ shadowed interference-limited scenario is given in \cite[Eqn.(6)]{kumar2017outage} \footnote{Note that there was an typo in the equation in the original version of \cite{kumar2017outage} and the correct expressions are used above. An errata for \cite{kumar2017outage} has also been communicated.}. An equivalent expression for the above CDF in the form of an infinite sum of Lauricella functions of the fourth kind is given below \cite[Eqn. (20)]{kumar2017outage}. We assume that the source and interferers undergo $\kappa-\mu$ shadowed fading with parameters $\mathlarger{(\kappa,\mu,m,\bar{x})}$ and $\mathlarger{\{(\kappa_i,\mu_i,m_i,\bar{y_i});i=1,\cdots,N\}}$ respectively, where $N$ is the number of interferers and $\bar{x},\{\bar{y_i};i=1,\cdots,N\}$ are the expectations of the corresponding RVs. Then the CDF of $\gamma$ is given by

\begin{equation}
\begin{aligned}
	F_{\gamma}(z) = & 1 - K_1 \sum_{p=0}^{\infty} \frac{(m)_p\left(1-\frac{\theta}{\lambda}\right)^p\Gamma\left[\sum\limits_{i=1}^{N}\mu_i+\mu+p\right]}{(\mu)_pp!} F_D^{(2N)}(1-p-\mu,\mu_1-m_1,\cdots,\mu_N-m_N,\\
	& m_1,\cdots,m_N;1+\sum_{i=1}^{N}\mu_i;\frac{\theta}{\theta+z\theta_1},\cdots,\frac{\theta}{\theta+z\theta_N},\frac{\theta}{\theta+z\lambda_1},\cdots,\frac{\theta}{\theta+z\lambda_N}),
\end{aligned}
\label{cdf1}
\end{equation}
where $\mathlarger{K_1=\frac{\prod\limits_{i=1}^{N}\left(\left(\frac{\theta}{\theta+z\theta_i}\right)^{\mu_i-m_i}\left(\frac{\theta}{\theta+z\lambda_i}\right)^{m_i}\right)\theta^m}{\Gamma\left[\sum\limits_{i=1}^{N}\mu_i+1\right]\Gamma[\mu]\lambda^m}}$, $\mathlarger{(m)_p=\frac{\Gamma[m+p]}{\Gamma[m]}}$ is the Pochhammer symbol and 
\begin{equation}
    \mathlarger{\theta = \frac{\bar{x}}{\mu_(1+\kappa)}}, \ \mathlarger{\theta_i = \frac{\bar{y_i}}{\mu_i(1+\kappa_i)}}, \  \mathlarger{\lambda = \frac{(\mu\kappa+m)\bar{x}}{\mu_(1+\kappa_)}, \lambda_i = \frac{(\mu_i\kappa_i+m_i)\bar{y_i}}{\mu_i(1+\kappa_i)m_i}}.
    \label{thetalambda}
\end{equation}
	To show that $F_{\gamma}(z)$ belongs to the \textit{MDA} of the Frechet distribution, we first do some simplifications of the CDF in (\ref{cdf1}). This CDF expression has a Lauricella function of the fourth kind ($F_D^{(2N)}(.)$), which has the following series expansion \cite{exton1976multiple}:
	\begin{equation}
		F_D^{(N)}(a,b_1,\cdots,b_N;c;x_1,\cdots,x_N) = \sum_{p_1,\cdots,p_N=0}^{\infty}\frac{(a)_{p_1+\cdots,p_N}}{(c)_{p_1+\cdots+p_N}}\prod_{i=1}^{N}(b_i)_{p_i}\frac{x_i^{p_i}}{p_i!}.
		\label{fd_expand}
	\end{equation}
	Substituting (\ref{fd_expand}) in (\ref{cdf1}), we get an expanded form for the CDF of $\gamma$ as given in (\ref{ccdf1}). Further, by rewriting the inner $2N$ fold summation in (\ref{ccdf1}) as two separate terms one with $p_1=\cdots=p_{2N}=0$ and the second with rest of the terms,we get (\ref{ccdf1_sum_expand}).
	
	\begin{equation}
	\begin{aligned}
	&F_{\gamma}(z)= 	1-\frac{\prod\limits_{i=1}^{N}\left(\frac{\theta}{\theta+z\theta_i}\right)^{\mu_i-m_i}\left(\frac{\theta}{\theta+z\lambda_i}\right)^{m_i}\theta^m}{\Gamma\left[\sum\limits_{i=1}^{N}\mu_i+1\right]\Gamma[\mu]\lambda^m}\times \sum_{p=0}^{\infty}\frac{(m)_p\left(1-\frac{\theta}{\lambda}\right)^p}{(\mu)_pp!}\Gamma\left[\sum_{i=1}^{N}\mu_i+p+\mu\right] \times\\
	& \sum_{p_1,\cdots,p_{2N}=0}^{\infty}\frac{(1-p-\mu)_{p_{1}+\cdots+p_{2N}}\prod\limits_{i=1}^{N}\left(\mu_i-m_i\right)_{p_{i}}(m_i)_{p_{i+N}}}{\left(1+\sum\limits_{i=1}^{N}\mu_i\right)_{p_1+\cdots+p_{2N}}}\prod_{i=1}^{N}\frac{\left(\frac{\theta}{\theta+z\theta_i}\right)^{p_i}\left(\frac{\theta}{\theta+z\lambda_i}\right)^{p_{i+N}}}{p_i!p_{i+N}!}.
	\end{aligned}
	\label{ccdf1}
	\end{equation}

	\begin{equation}
	\begin{aligned}
    	&1-F_{\gamma}(z) = K_2\sum_{p=0}^{\infty}C_1\left \lbrace C_2z^{-\sum\limits_{i=1}^{N}\mu_{i}}\prod_{i=1}^{N}\left(\frac{\theta}{z}+\theta_i\right)^{-(\mu_{i}-m_i)}\left(\frac{\theta}{z}+\lambda_i\right)^{-m_i}+ \right. \\ & \left. \sum\limits_{\underset{\mathlarger{s.t \ \exists \ p_{i_1}\neq0;
	\ i_1 \in \{1,\cdots,2N\}}}{p_1\cdots p_{2N}=0 ,}  }^{\infty}C_2z^{-\sum\limits_{i=1}^{N}\mu_{i}}z^{-\sum\limits_{i=1}^{N}(p_i+p_{i+N})}\prod\limits_{i=1}^{N}\frac{\left(\frac{\theta}{z}+\theta_i\right)^{m_i-\mu_i-p_i}\left(\frac{\theta}{z}+\lambda_i\right)^{-m_i-p_{i+N}}}{p_i!p_{i+N}!}\right \rbrace,
		\end{aligned}
		\label{ccdf1_sum_expand}
	\end{equation}
 where $\mathlarger{K_2=\frac{\left(\sfrac{\theta}{\lambda}\right)^m}{\Gamma\left[1+\sum\limits_{i=1}^{N}\mu_i\right]\Gamma
		[\mu]}}$, $\mathlarger{C_1=\frac{(m)_p\left(1-\frac{\theta}{\lambda}\right)^p}{(\mu)_pp!}\Gamma\left[\sum\limits_{i=1}^{N}\mu_i+p+\mu\right]}$ and \\ $\mathlarger{C_2=\frac{(1-p-\mu)_{p_{1}+\cdots+p_{2N}}\prod\limits_{i=1}^{N}\left(\mu_i-m_i\right)_{p_{i}}(m_i)_{p_{i+N}}}{\left(1+\sum\limits_{i=1}^{N}\mu_i\right)_{p_1+\cdots+p_{2N}}}\theta^{\sum\limits_{i=1}^{N}\mu_i+p_i+p_{i+N}}}$. After further rearrangement of the terms, we obtain,
	\begin{equation}
	\begin{aligned}
	1-F_{\gamma}(z) & = K_2\sum_{p=0}^{\infty}C_1 \left\lbrace C_2z^{-\sum\limits_{i=1}^{N}\mu_{i}}\prod_{i=1}^{N}\left(\frac{\theta}{z}+\theta_i\right)^{-(\mu_{i}-m_i)}\left(\frac{\theta}{z}+\lambda_i\right)^{-(m_i)}+ 
	z^{-\sum\limits_{i=1}^{N}\mu_{i} }h(z)\right\rbrace,
	\end{aligned}
	\label{ccdf5}
	\end{equation}
	where $\mathlarger{h(z)=\sum\limits_{\underset{\mathlarger{s.t \ \exists \ p_{i_1}\neq0;
	\ i_1 \in \{1,\cdots,2N\}}}{p_1\cdots p_{2N}=0 ,}  }^{\infty}C_2z^{-\sum\limits_{i=1}^{N}(p_i+p_{i+N})}\prod\limits_{i=1}^{N}\frac{\left(\frac{\theta}{z}+\theta_i\right)^{m_i-\mu_i-p_i}\left(\frac{\theta}{z}+\lambda_i\right)^{-m_i-p_{i+N}}}{p_i!p_{i+N}!}}$.\\ Now, let us focus on the term $h(z)$. As $\mathlarger{z\to\infty}$,  $\mathlarger{z^{-\sum\limits_{i=1}^{N}(p_i+p_{i+N})}}$ tends to zero and the product $\mathlarger{\prod\limits_{i=1}^{N}\frac{\left(\frac{\theta}{z}+\theta_i\right)^{m_i-\mu_i-p_i}\left(\frac{\theta}{z}+\lambda_i\right)^{-m_i-p_{i+N}}}{p_i!p_{i+N}!}}$ will tend to a finite and positive value, which is  $\mathlarger{\frac{\prod\limits_{i=1}^{N}\theta_i^{m_i-\mu_i-p_i}\lambda_i^{-m_i-p_{i+N}}}{p_i!p_{i+N}!}}$. Hence $\mathlarger{\lim\limits_{z\to\infty} h(z)=0}$. 
		Recall from (\ref{codn_frec1}) that the condition for a CDF $F$ to belong to the MDA for Frechet distribution is 
		 \begin{equation}
	 	\lim_{t\to\infty}\frac{1-F(tz)}{1-F(t)} \ = \ z^{-\beta}.
	 	\label{codn_frec}
	 \end{equation} 
Substituting (\ref{ccdf5}) in the left hand side (LHS) of the above relation, we have,
\begin{equation}
\lim_{t\to\infty} \frac{K_2\sum\limits_{p=0}^{\infty}C_1\{C_2(tz)^{-\sum\limits_{i=1}^{N}\mu_{i}}\prod\limits_{i=1}^{N}(\frac{\theta}{tz}+\theta_i)^{-(\mu_{i}-m_i)}(\frac{\theta}{tz}+\lambda_i)^{-(m_i)}+ 
	(tz)^{-\sum\limits_{i=1}^{N}\mu_{i} }h(tz) \}}{K_2\sum\limits_{p=0}^{\infty}C_1\{C_2(t)^{-\sum\limits_{i=1}^{N}\mu_{i}}\prod\limits_{i=1}^{N}(\frac{\theta}{t}+\theta_i)^{-(\mu_{i}-m_i)}(\frac{\theta}{t}+\lambda_i)^{-(m_i)}+ 
	(t)^{-\sum\limits_{i=1}^{N}\mu_{i} }h(t) \}}.
\end{equation}
Cancelling the terms common to both numerator and denominator, we obtain,  
\begin{equation}
\begin{aligned}
\lim_{t\to\infty} & \frac{z^{-\sum\limits_{i=1}^{N}\mu_{i}} K_2\sum\limits_{p=0}^{\infty}C_1\{C_2\prod\limits_{i=1}^{N}(\frac{\theta}{tz}+\theta_i)^{-(\mu_{i}-m_i)}(\frac{\theta}{tz}+\lambda_i)^{-(m_i)}+ 
	h(tz) \}}{K_2\sum\limits_{p=0}^{\infty}C_1\{C_2\prod\limits_{i=1}^{N}(\frac{\theta}{t}+\theta_i)^{-(\mu_{i}-m_i)}(\frac{\theta}{t}+\lambda_i)^{-(m_i)}+ 
	h(t) \}}.
	\label{mda_simpligy}
\end{aligned}
\end{equation}
 Since we have already proven that $\mathlarger{\lim\limits_{z\to\infty}}$ $\mathlarger{h(z)=0}$, (\ref{mda_simpligy}) evaluates to $\mathlarger{z^{-\sum\limits_{i=1}^{N}\mu_{i}}}$.

	\section{Derivation of rate of convergence} \label{rateofconv}
	To derive the result in Theorem \ref{rate_cnvg}, we first define the $\delta$-neighborhood of GPD for a Frechet RV. Let the $\delta$-neighbourhood be denoted by $Q_1(\delta)$ and the GPD for a Frechet RV be denoted by $W_{\{1,\beta\}}$. The Extreme Value Distributions (EVDs) lies in the $\delta$ neighbourhood of one of three GPD $W_{\{i,\beta\}}; \  i=1,2,3$ with $\delta=1$ .
	
\theoremstyle{definition}
	\begin{defn}$\delta$-neighborhood $Q_1(\delta)$ of the GPD  $W_{\{1,\beta\}}$ \cite{falk2010laws} is defined as
		$Q_1(\delta)$ := \{F\ : \ $\omega$(F)\ =\ $\infty\}$ and $F$ has a density $f$ on $[z_0,\infty]$ for some $z_0>0$ such that for some shape parameter $\beta >0$ and some scale parameter $a>0$  on  $[z_0,\infty]$, we have,
		\begin{equation}
		f(z)=\frac{1}{a}W'_{1,\beta}\left(\frac{z}{a}\right)(1+O((1-W_{1,\beta}(z))^\delta) \},
		\label{pdf_condition}
		\end{equation}  
		where $\mathlarger{\omega(F):= sup \{z \in \mathbb{R} : F(z)<1 \}}$.
		In fact the GPD for the Frechet distribution is defined in \cite{falk2010laws} as $\mathlarger{W_{1,\beta} = 1 - z^{-\beta}; z \geq1}$ and using this, (\ref{pdf_condition}) can be rewritten as 
	\begin{equation}
	f(z) = \frac{\beta}{a}\left( \frac{z}{a}\right)^{-\beta-1}\left( 1+O((z^{-\beta})^\delta)\right).
	\label{fz_with_gpd}
	\end{equation}
	\end{defn}
		This definition says that, if a PDF $f$ on  $[z_0,\infty]$ for some $z_0>0$ can be written in the form of (\ref{fz_with_gpd}), then the corresponding CDF $F$ belongs to the $\delta$-neighborhood $Q_{1}(\delta)$ of the Frechet distribution\footnote{For a  real or complex valued function $g_1(x)$ and a strictly positive real valued function $g_2(x)$ both defined on some unbounded subset of $\mathbb{R}^+$, we say $g_1(x)=O(g_2(x))$, iff $\exists$ $M\in \mathbb{R}^+$ and $x_0 \in \mathbb{R}$ such that, $|g_1(x)| \leq Mg_2(x) $ $\forall x \geq x_0$.}. The PDF of the SIR RV at the $j^{th}$ antenna is given by,
		
\begin{equation}
\begin{aligned}
	f_{\gamma}(z) =  K_5 z^{-\left(1+\sum\limits_{i=1}^{N}\mu_{i}\right)} \left(1+\frac{\theta}{z \theta_1}\right)^{-\left(\mu+\sum\limits_{i=1}^{N}\mu_{i}\right)}\times _{(1)}^{(1)}E_D^{(2N)}\left[\mu+\sum_{i=1}^{N}\mu_{i},m,\mu_2-m_2,\cdots, \right.\\ \left. \mu_N-m_N,m_1,\cdots,m_N;\mu,\sum_{i=1}^{N}\mu_{i};
	\frac{z\theta_1(\lambda-\theta)}{\lambda(\theta+z\theta_1)},\frac{\theta(\theta_2-\theta_1)}{\theta_2(\theta+z\theta_1)}, \cdots,\frac{\theta(\theta_N-\theta_1)}{\theta_N(\theta+z\theta_1)}, \frac{\theta(\lambda_1-\theta_1)}{\lambda_1(\theta+z\theta_1)}, \right. \\ \left. \cdots, \frac{\theta(\lambda_N-\theta_1)}{\lambda_N(\theta+z\theta_1)}\right],
\end{aligned} 
\label{pdf}
\end{equation}
 where $\mathlarger{K_5 = \frac{\theta^{(m+\sum\limits_{i=1}^{N}\mu_i)}\Gamma\left[\mu+\sum\limits_{i=1}^{N}\mu_i\right]}{\lambda^m \Gamma[\mu] \Gamma\left[\sum\limits_{i=1}^{N}\mu_i\right]\prod\limits_{i=1}^{N}\theta_i^{\mu_i-m_i}\lambda_i^{m_i}}}$
(The PDF expression is not given in \cite{kumar2017outage} and is a non-trivial derivation. Hence, we derive the PDF in Appendix \ref{appendixb}).
	Also, we get the following form for the PDF $\mathlarger{f_{\gamma}(z)}$ by following the simplification steps given in Appendix \ref{appendixc}: 
	\begin{equation}
	f_{\gamma}(z) = K_6 z^{-(1+\sum\limits_{i=1}^{N}\mu_i)} \left(1+\frac{(\mu+\sum\limits_{i=1}^{N}\mu_i)m(\lambda-\theta)}{\mu\lambda}\right) (1+O(z^{-1})). 
	\label{fz_sir1}
	\end{equation}
	This is in the same form as that of (\ref{fz_with_gpd}). 
	Comparing (\ref{fz_sir1}) with (\ref{fz_with_gpd}), we can identify that the CDF $\mathlarger{F_{\gamma}(z)}$ belongs to the $\delta$ neighborhood of $Q_{1}(\delta)$ with $\mathlarger{\delta = \left(\sum\limits_{i=1}^{N}\mu_{i} \right)^{-1}}$ and $\mathlarger{\beta=\sum\limits_{i=1}^{N}\mu_i}$.	Now that we have identified the $\delta$ neighbourhood for $\mathlarger{F_{\gamma}(z)}$, we make use of the following lemma from \cite{falk2010laws} to conclude the proof. 
	
		\begin{lemma}
		Suppose that the CDF $F$ (of i.i.d. RVs $z_1,\cdots,z_L$) is in the $\delta$ neighborhood $Q_1(\delta)$ of the GPD $W_{1,\beta}$  then there obviously exist constant $a > 0$ such that  $\mathlarger{f(z)=\frac{1}{a}W'_{1,\beta}(\frac{z}{a})(1+O((1-}$ $\mathlarger{W_{1,\beta}(z))^\delta)} $ for all $z$ in the left neighborhood of $\mathlarger{\omega(W_{1,\beta})}$. Consequently we have, 
		
		\begin{equation}
		\sup_{B\in \mathbb{B}} \left| \mathbb{P}\left(\left(\left( \frac{M_L}{a}  \right) / L^{\beta}\right) \in B\right) - G_{\beta}(B)\right| = O\left(\left(\frac{1}{L}\right)^{\delta} + \frac{1}{L}\right),
		\end{equation}
		
	where $\mathbb{B}$ denotes the Borel $\sigma$ algebra on $\mathbb{R}$ and $M_L = max \{z_1,\cdots,z_L \} $.	
	\end{lemma}
	 Since the CDF $F_{\gamma}(z)$ belongs to the $\delta$ neighborhood of $Q_1(\delta)$, by the previous lemma, the rate of convergence is $O\left(\left(\frac{1}{L}\right)^{\delta} + \frac{1}{L}\right)$ with $\delta = \left(\sum\limits_{i=1}^{N}\mu_{i} \right)^{-1}$ .
\section{Derivation of PDF of SIR random variable} \label{appendixb}
	 Let $Z=\frac{X}{Y}$, $Y=\sum\limits_{i=1}^{N}Y_i$. Here, $X$ and $Y_i$ are $\kappa-\mu$ shadowed RVs with parameters $\{\kappa,\mu,m,\bar{x}\}$ and $\{\kappa_i,\mu_i,m_i,\bar{y_i}\}$ respectively. Then by the method of transformation of RVs, the PDF of $Z$ can be expressed as $f_Z(z) = \int\limits_{0}^{\infty} yf_X(yz)f_Y(y) \ dy,$ where $f_X(x)$ and $f_Y(y)$ represent the PDFs of $X$ and $Y$ respectively. The expression for the PDF of $X$ is given in (\cite[Eqn. (4)]{paris2014statistical}). The PDF of sum of i.n.i.d.  $\kappa-\mu$ shadowed RVs is given in \cite{paris2014statistical} as follows, 
	 
\begin{equation}
\begin{aligned}
f_Y(y)	 = \frac{y^{\sum\limits_{i=1}^{N}\mu_{i}-1}}{\Gamma\left[\sum\limits_{i=1}^{N}\mu_{i}\right]\prod\limits_{i=1}^{N}\left(\theta_i^{\mu_{i}-m_i} \lambda_i^{m_i}\right)}\phi_2^{(2N)}\left(\mu_1-m_1,\cdots,\mu_N-m_N,m_1,\cdots,m_N; \right. \\ \left.
\sum_{i=1}^{N}\mu_{i};-\frac{y}{\theta_1},\cdots,-\frac{y}{\theta_N},-\frac{y}{\lambda_1},\cdots,-\frac{y}{\lambda_N}\right),
\end{aligned}
\end{equation} 
where $\mathlarger{\theta_i = \frac{\bar{y_i}}{\mu_i(1+\kappa_i)}}$, $\mathlarger{\lambda_i = \frac{(\mu_i\kappa_i+m_i)\bar{y_i}}{\mu_i(1+\kappa_)m_i}}$ for $\mathlarger{i=1,\cdots,N}$ and $\mathlarger{\phi_2^{(2N)}(.)}$ is the confluent multivariate hypergeometric function of $2N$ variables. Substituting the pdfs of $X$ and $Y$ in the expression for $f_Z(z)$, we obtain,

\begin{equation}
	\begin{aligned}
	f_Z(z)& =  \ K_3 \int\limits_{0}^{\infty}y^{\mu+\sum\limits_{i=1}^{N}\mu_{i}-1} e^{-\frac{yz}{\theta}} \times _1F_1\left(m,\mu,\frac{yz}{\theta}-\frac{yz}{\lambda}\right) \times \\ 
	&\phi_2^{(2N)}\left(\mu_1-m_1,\cdots,\mu_N-m_N,m_1,\cdots,m_N;
	\sum_{i=1}^{N}\mu_{i};-\frac{y}{\theta_1},\cdots,-\frac{y}{\theta_N},-\frac{y}{\lambda_1},\cdots,-\frac{y}{\lambda_N}\right),
	\end{aligned}
	\label{int_pdf}
\end{equation}
where $\mathlarger{K_3 = \frac{z^{\mu-1}}{\theta^{\mu-m}\lambda^m\Gamma[\mu]\Gamma\left[\sum\limits_{i=1}^{N}\mu_{i}\right]}\frac{1}{\prod\limits_{i=1}^{N}\theta_i^{\mu_{i}-m_i} \lambda_i^{m_i} }}$. We use the following integral identity from \cite{exton1976multiple} to simplify the integral expression  in (\ref{int_pdf}):

\begin{equation}
	\begin{aligned}
	\Gamma[a]_{(1)}^{(1)}E_D^{(N)}[a,b_1,\cdots,b_N;c,c';x_1,\cdots,x_N] = \int\limits_{0}^{\infty}e^{-t}t^{a-1}\phi_2^{(k)}[b_1,\cdots,b_k;c;x_1t,\cdots,x_kt] \\ \times \phi_2^{(N-k)}[b_{k+1},\cdots,b_{N};c';x_{k+1}t,\cdots,x_Nt] \ dt.
	\end{aligned}
	\label{ed_int}
\end{equation}

Now, the PDF of $Z$ is given by,
\begin{equation}
\begin{aligned}
	f_Z(z) = K_4 \times_{(1)}^{(1)}E_D^{(2N+1)}\left[\mu+\sum_{i=1}^{N}\mu_{i},m,\mu_1-m_1,\cdots,\mu_N-m_N,m_1, \cdots,m_N;\mu,\sum_{i=1}^{N}\mu_{i}; \right. \\ \left. 1-\frac{\theta}{\lambda},-\frac{\theta}{z\theta_1}, \cdots,-\frac{\theta}{z\theta_N},-\frac{\theta}{z\lambda_1},\cdots,-\frac{\theta}{z\lambda_N}\right],
\end{aligned}
\end{equation}
where $\mathlarger{K_4 = K_3\Gamma\left[\mu+\sum\limits_{i=1}^{N}\mu_{i}\right]}$.
Now, we make use of the following transformation from  \cite{exton1976multiple} to get a converging form for the PDF: $^{(1)}_{(1)}E^{(N)}_{D}(a,b_1,\cdots,b_N;c,c';x_1,\cdots,x_N)=(1-x_2)^{-a} \times ^{(1)}_{(1)}E^{(N)}_{D}(a,b_1,c'-b_{2}-\cdots-b_N,b_{3},\cdots,b_N;c,c';\frac{x_1}{1-x_2},\frac{x_2}{x_2-1},\frac{x_2-x_3}{x_3-1},\cdots,\frac{x_2-x_N}{x_2-1}).$

Hence, we obtain the following expression:
\begin{equation}
\begin{aligned}
f_Z(z) =  K_5 z^{-(1+\sum\limits_{i=1}^{N}\mu_{i})} \left(1+\frac{\theta}{z \theta_1}\right)^{-\left(\mu+\sum\limits_{i=1}^{N}\mu_{i}\right)}\times _{(1)}^{(1)}E_D^{(2N)}\left[\mu+\sum_{i=1}^{N}\mu_{i},m,\mu_2-m_2,\cdots, \right.\\ \left. \mu_N-m_N,m_1,\cdots,m_N;\mu,\sum_{i=1}^{N}\mu_{i};
\frac{z\theta_1(\lambda-\theta)}{\lambda(\theta+z\theta_1)},\frac{\theta(\theta_2-\theta_1)}{\theta_2(\theta+z\theta_1)}, \cdots,\frac{\theta(\theta_N-\theta_1)}{\theta_N(\theta+z\theta_1)}, \frac{\theta(\lambda_1-\theta_1)}{\lambda_1(\theta+z\theta_1)}, \right. \\ \left. \cdots, \frac{\theta(\lambda_N-\theta_1)}{\lambda_N(\theta+z\theta_1)}\right],
\end{aligned} 
\label{pdf_exp}
\end{equation}
where $\mathlarger{K_5 = \frac{\theta^{(m+\sum\limits_{i=1}^{N}\mu_i)}\Gamma\left[\mu+\sum\limits_{i=1}^{N}\mu_i\right]}{\lambda^m \Gamma[\mu] \Gamma\left[\sum\limits_{i=1}^{N}\mu_i\right]\prod\limits_{i=1}^{N}\theta_i^{\mu_i-m_i}\lambda_i^{m_i}}}$.

\section{Simplification of PDF to identify $\delta$  neighbourhood} \label{appendixc}
To begin with, the PDF $f_{\gamma}(z)$ given by (\ref{pdf}) is rewritten as in Eqn. \ref{pdf_sir} where $\mathlarger{z_1 = \frac{(\lambda-\theta)z\theta_1}{\lambda(z\theta_1+\theta)}}$, $\mathlarger{z_i=\frac{\theta(\theta_i-\theta_1)}{\theta_i(\theta+z\theta_1)}}$ for $\mathlarger{i \in \{2,\cdots N\}}$ and $\mathlarger{z_i=\frac{\theta(\lambda_i-\theta_1)}{\lambda_i(\theta+z\theta_1)}}$ for  $\mathlarger{i \in \{N+1,\cdots, 2N\}}.$ The $\mathlarger{E_D^{(2N)}(.)}$ term in this expression has the following series expansion from \cite{exton1976multiple}: 
	\begin{equation}
	    _{(1)}^{(1)}E_D^{(N)}[a,b_1,\cdots,b_N;c,c';x_1,\cdots,x_N] = \sum\limits_{p_1,\cdots,p_N=0}^{\infty}\frac{(a)_{p_1+\cdots+p_N}\prod\limits_{i=1}^{N}(b_i)_{p_i}\prod\limits_{i=1}^{N}x_i^{p_i}}{(c)_{p_1}(c')_{p_2+\cdots+p_N}p_1!\cdots p_N!}.
	\end{equation}
	\begin{equation}
	\begin{aligned}
	&	f_{\gamma}(z) = \frac{z^{-\left(1+\sum\limits_{i=1}^{N}\mu_i\right)}\theta^{\left(m+\sum\limits_{i=1}^{N}\mu_i \right)}\Gamma\left[\mu+\sum\limits_{i=1}^{N}\mu_i\right]}{\lambda^m \Gamma[\mu] \Gamma\left[\sum\limits_{i=1}^{N}\mu_i\right]\prod\limits_{i=1}^{N}\theta_i^{\mu_i-m_i}\lambda_i^{m_i}}\left(1+\frac{\theta}{z\theta_1}\right)^{-\left(\mu+\sum\limits_{i=1}^{N}\mu_i\right)}  \\
	& \times_{(1)}^{(1)}E_D^{(2N)}\left(\mu+\sum_{i=1}^{N}\mu_i,m,\mu_2-m_2,\cdots,\mu_N-m_N,m_1,\cdots,m_N;\mu,\sum_{i=1}^{N}\mu_i;z_1,\cdots,z_{2N}\right).
	\end{aligned}
	\label{pdf_sir}
	\end{equation}
	
	Using the above series expansion, we rewrite (\ref{pdf_sir}) as 
	\begin{equation}
	\begin{aligned}
	f_{\gamma}(z) = & K_6z^{-\left(1+\sum\limits_{i=1}^{N}\mu_i\right)} \left(1+\frac{\theta}{z\theta_1}\right)^{-\left(\mu+\sum\limits_{i=1}^{N}\mu_i\right)}\sum\limits_{p_1,\cdots,p_{2N}=0}^{\infty}\frac{\left(\mu+\sum\limits_{i=1}^{N}\mu_i\right)_{p_1+\cdots+p_{2N}}}{(\mu)_{p_1}}  \\ 
	&  \times \frac{(m)_{p_1}  \prod\limits_{i=2}^{N}(\mu_i-m_i)_{p_i}  \prod\limits_{i=N+1}^{2N}(m_i)_{p_i} }{\left(\sum\limits_{i=1}^{N}\mu_i\right)_{p_2+\cdots+p_{2N}}}\prod_{i=1}^{2N}\frac{z_i^{p_i}}{p_i!},
	\end{aligned}
	\label{pdf_series}
	\end{equation}
	where $\mathlarger{K_6 := \frac{\theta^{\left(m+\sum\limits_{i=1}^{N}\mu_i\right)}\Gamma\left[\mu+\sum\limits_{i=1}^{N}\mu_i\right]}{\lambda^m \Gamma[\mu] \Gamma\left[\sum\limits_{i=1}^{N}\mu_i\right]\prod\limits_{i=1}^{N}\theta_i^{\mu_i-m_i}\lambda_i^{m_i}}}$.
We then expand the $2N$ fold summation in (\ref{pdf_series}) into three terms: the first term with all the iterating variables $\mathlarger{p_1, p_2,..., p_{2N}}$ taking the value zero, the second term with exactly one non-zero iterating variable and the third term with the rest. By expanding, (\ref{pdf_series}) becomes the expression given in (\ref{fz_expand}) where $\rho = \mu+\sum\limits_{i=1}^{N}\mu_i$. The term $\mathlarger{\frac{1}{\theta+z\theta_1}}$ present in $\mathlarger{Term \ a}$ and $\mathlarger{Term \ b}$ of (\ref{fz_expand}) has the following converging series expansion: 
	\begin{equation}
	\frac{1}{z\theta_1+\theta} = \frac{1}{\theta} \left\lbrace \frac{1}{z\theta_1/\theta} - \frac{1}{(z\theta_1/\theta)^2} + \frac{1}{(z\theta_1/\theta)^3} - \frac{1}{(z\theta_1/\theta)^4} + O(z^{-5})\right\rbrace .
	\label{series1}
	\end{equation}
	Using (\ref{series1}), $Term \ a$  can be represented as 
	\begin{equation}
	\frac{(\mu+\sum\limits_{i=1}^{N}\mu_i)m(\lambda-\theta)}{\mu\lambda} \left\lbrace 1 -  \frac{1}{z\theta_1/\theta} + \frac{1}{(z\theta_1/\theta)^2} - \frac{1}{(z\theta_1/\theta)^3} + O(z^{-4}) \right\rbrace.
	\end{equation}
	Similarly, $Term \ b$ can also be expanded to get a series expression, but the expansion will have only negative powers of $z$. $Term \ 3$ will also have only powers of $z$ less than 1. Combining these series expansions, the SIR PDF can be finally expressed as, 
	\begin{equation}
	f_{\gamma}(z) = K_6 z^{-(1+\sum\limits_{i=1}^{N}\mu_i)} \left(1+\frac{(\mu+\sum\limits_{i=1}^{N}\mu_i)m(\lambda-\theta)}{\mu\lambda}\right) (1+O(z^{-1})). 
	\label{fz_sir}
	\end{equation}
	
\begin{equation}
\begin{aligned}
f_{\gamma}(z) = & \  \left(K_6z^{-\left(1+\sum\limits_{i=1}^{N}\mu_i\right)} \left(1+\frac{\theta}{z\theta_1}\right)^{-\rho}\right) \left\lbrace  \underbrace{1}_\textit{Term 1}+  \right. \\ &  \left.   \underbrace{\underbrace{\frac{\rho m (\lambda-\theta)z\theta_1}{\mu \lambda(z\theta_1 + \theta)}}_\textit{Term a} + \sum\limits_{k=2}^{N} \underbrace{ \frac{\rho(\mu_k-m_k) \theta(\theta_k-\theta_1)}{\sum\limits_{i=1}^{N}\mu_i \theta_k(\theta+z\theta_1)}+\sum\limits_{k=N+1}^{2N}\frac{\rho(m_k)\theta(\lambda_k-\theta_1)}{\sum\limits_{i=1}^{N}\mu_i \lambda_k(\theta+z\theta_1)}}_\textit{Term b} \ }_\textit{Term 2}    \right. \\ & \left.   + \underbrace{  \sum\limits_{\underset{\mathlarger{ \exists \ i_1,i_2 s.t \ p_{i_1}p_{i_2}\neq0 \ \forall \ i_1\neq i_2}}{p_1,\cdots,p_{2N}=0;}}^{\infty} \frac{(\rho)_{p_1+\cdots+p_{2N}}}{(\mu)_{p_1}}\frac{(m)_{p_1}  \prod\limits_{i=2}^{N}(\mu_i-m_i)_{p_i}  \prod\limits_{i=N+1}^{2N}(m_i)_{p_i} \prod\limits_{i=1}^{2N}\frac{z_i^{p_i}}{p_i!} }{\left(\sum\limits_{i=1}^{N}\mu_i\right)_{p_2+\cdots+p_{2N}}}}_\textit{Term 3} \right\rbrace .
\end{aligned}
\label{fz_expand}
\end{equation}

\end{appendices}
\bibliographystyle{IEEEtran}
\bibliography{library}
\end{document}